%% file: main.tex
\documentclass[11pt,letterpaper]{article}
\usepackage{amsmath,amsfonts,amsthm,amssymb}
\usepackage{setspace}
\usepackage{fancyhdr}
\usepackage{lastpage}
\usepackage{extramarks}
\usepackage{xspace}
\usepackage{chngpage}
\usepackage[ruled,linesnumbered,vlined]{algorithm2e}
\usepackage{soul,color}
\usepackage{graphicx,float,wrapfig}
\usepackage[font=small,labelfont=bf]{caption}
\usepackage[margin=1in]{geometry}
\usepackage{thmtools,thm-restate}
\usepackage{verbatim}
\usepackage[linesnumbered,ruled,vlined]{algorithm2e}

\usepackage[nice]{nicefrac}

\usepackage{mathrsfs}

\allowdisplaybreaks

\usepackage{mdframed}
\newtheorem*{mdresult}{Result}

\input{notations}

\usepackage[usenames,dvipsnames]{xcolor}
\usepackage[linktocpage=true,breaklinks,colorlinks,citecolor=blue,linkcolor=BrickRed]{hyperref}
\usepackage{cleveref}
\newcommand{\algref}[1]{\hyperref[{#1}]{Algorithm~\ref*{#1}}}
\newcommand{\IGNORE}[1]{}

\usepackage{tikz}
\usetikzlibrary{graphs, graphs.standard}
\usetikzlibrary{arrows.meta}

\usepackage{tcolorbox}
\usepackage{enumerate}

\title{Online Stochastic Matching with Unknown Arrival Order: \\
Beating $0.5$ against the Online Optimum}

\author{
Enze Sun\thanks{The University of Hong Kong (\texttt{sunenze@connect.hku.hk}). Supported by NSFC6212290003.}
\and
Zhihao Gavin Tang\thanks{ITCS, Key Laboratory of Interdisciplinary Research of Computation and Economics, Shanghai University of Finance and Economics (\texttt{tang.zhihao@mail.shufe.edu.cn}). Supported by National Key R\&D Program of China (2023YFA1009500).}
\and
Yifan Wang\thanks{School of Computer Science, Georgia Tech (\texttt{ywang3782@gatech.edu}). Supported in part by NSF awards CCF-2327010 and CCF-2440113.}
}
\date{}

\newtheorem{Theorem}{Theorem}[section]
\newtheorem{Lemma}[Theorem]{Lemma}

\newtheorem{Observation}[Theorem]{Observation}

\newtheorem{Claim}[Theorem]{Claim}

\newtheorem{Corollary}[Theorem]{Corollary}

\newtheorem*{Lemma*}{Lemma}

\begin{document}

\maketitle 

\begin{abstract}
\input{abstract}
\end{abstract}

\input{intro}

\input{prelim}

\input{warm_up}

\input{preprocessing}

\input{slackness}

\input{small_slack}

\input{large_slack}

\newpage

\bibliographystyle{alpha}
\bibliography{ref.bib,bib.bib}

\newpage

\appendix

\input{smallslack_appendix}

\end{document}

%% file: notations.tex
\DeclareMathOperator*{\argmax}{arg\,max}

\newcommand{\stlb}{\mathrm{LB}}
\newcommand{\algbase}{\mathrm{Alg^{b}}}
\newcommand{\algbi}{\mathrm{Alg}^{\mathrm{b}}_i}

\newcommand{\mwm}{{\rm{MWM}}}

\newcommand{\lp}{\mathrm{LP}}
\newcommand{\feasib}{\mathcal{P}}
\newcommand{\lpexante}{\mathrm{LP_{ex-ante}}}

\newcommand{\lpcore}{\mathrm{LP_{slack}}}
\newcommand{\lpslack}{\mathrm{LP_{slack}}}

\newcommand{\lpsubmod}{\mathsf{LP_{submod}}}

\newcommand{\philo}{\mathrm{OPT_{online}}}

\newcommand{\deltnew}{\delta_0}

\newcommand{\probfr}{{\delta_x}}

\newcommand{\epsalg}{\varepsilon_{\mathsf{Alg}}}

\newcommand{\epsslack}{\varepsilon_{\mathrm{s}}}

\newcommand{\wit}[1][t]{w_{i#1}}

\newcommand{\xit}[1][t]{x_{i#1}}
\newcommand{\yit}[1][t]{y_{i#1}}

\newcommand{\vt}{t}
\newcommand{\ui}{i}

\newcommand{\eps}{\varepsilon}

\newcommand{\epsphilo}{\varepsilon_o}
\newcommand{\eqdef}{\overset{\mathrm{def}}{=\mathrel{\mkern-3mu}=}}
\newcommand{\tildxl}{{\tilde{x}^L}}
\newcommand{\tildyl}{{\tilde{y}^L}}
\newcommand{\tildx}{{\tilde{x}}}
\newcommand{\tildy}{{\tilde{y}}}
\newcommand{\tildz}{{\tilde{z}}}

\newcommand{\yifana}{\varepsilon_{\mathrm{s1}}}
\newcommand{\yifanb}{\varepsilon_{\mathrm{s2}}}

\newcommand{\fr}{\mathsf{FR}}

\newcommand{\dt}{\mathsf{DT}}

\newcommand{\alg}{\mathrm{ALG}}

\newcommand{\E}{\mathbb{E}}

\newcommand{\one}{\mathbf{1}}%
\renewcommand{\Pr}[2][]{\mbox{\rm\bf Pr}_{#1}\left[#2\right]}%

%% file: abstract.tex
We study the online stochastic matching problem. Against the offline benchmark, Feldman, Gravin, and Lucier (SODA 2015) designed an optimal $0.5$-competitive algorithm. A recent line of work, initiated by Papadimitriou, Pollner, Saberi, and Wajc (MOR 2024), focuses on designing approximation algorithms against the online optimum. The online benchmark allows positive results surpassing the $0.5$ ratio.

In this work, adapting the order-competitive analysis by Ezra, Feldman, Gravin, and Tang (SODA 2023), we design a $0.5+\Omega(1)$ order-competitive algorithm against the online benchmark with unknown arrival order. Our algorithm is significantly different from existing ones, as the known arrival order is crucial to the previous approximation algorithms.  


%% file: intro.tex
\section{Introduction}
\label{sec:intro}


We study the ``Bernoulli'' online stochastic bipartite matching problem.
Let there be an edge-weighted graph $G=(U\cup V,E,w)$.
One side of the graph $U$, known as offline vertices, is given upfront. Another side $V$ contains online vertices that arrive sequentially.
At each step, a vertex $t \in V$ (an online vertex) arrives and is realized (with possibly non-zero edge weights $w_{it}$) independently with probability $p_t$, and with probability $1 - p_t$ the vertex comes with edge weights $0$. The algorithm then decides immediately whether and how to match $t$ if $t$ is realized. The matching decisions are irrevocable and the goal is to maximize the total weight of the matching.
Notice that the probabilities $\{p_t\}$ and weights $\{w_{it}\}$ are known upfront, while the actual realizations of vertices are only observed upon arrival.
\footnote{A more general version of the problem assumes independent random types of online vertices, where the type of vertex $t$ determines the edge weights $\{\wit\}_{i}$.}

The problem is extensively studied in the literature and the conventional wisdom is to use the offline optimum (a.k.a. prophet) as the benchmark. A tight $0.5$ competitive ratio is established by Feldman, Gravin, and Lucier~\cite{FGL-SODA15} and is further generalized to the setting of online stochastic matching on non-bipartite graphs~\cite{EFGT-MOR22}. Restricted to the unweighted and vertex-weighted version of the problem, the optimal competitive ratio is in between $0.69$ and $0.75$~\cite{TWW-STOC22,CHS-arxiv24}. 

Recently, Papadimitriou, Pollner, Saberi, and Wajc~\cite{PPSW-MOR24} revisited the online stochastic matching problem by switching the benchmark from the offline optimum to the online optimum.
They focused on the computational perspective of the problem, proving that the online optimum is PSPACE-hard to approximate within a factor of $(1-\Omega(1))$. On the positive side, the authors designed an efficient $0.51$-approximation algorithm against the online optimum. The ratio is improved by a series of subsequent works~\cite{SW-icalp21,BDL-EC22,NSW-arxiv23,BDPSW-arxiv24}, and the state-of-the-art approximation ratio is $0.678$ by Braverman, Derakhshan, Pollner, Saberi, and Wajc~\cite{BDPSW-arxiv24}. 

Another line of work by Ezra, Feldman, Gravin, and Tang~\cite{EFGT-SODA23} proposed the concept of order-competitive ratio that measures the value of knowing the arrival order of an online problem. Specifically, the framework also considers the online optimum benchmark but aims at designing online algorithms with unknown arrival order. They focused on the single-choice setting and designed optimal deterministic order-competitive algorithms. The results are improved through randomized algorithms~\cite{CST-arxiv24}, and are generalized to combinatorial settings~\cite{EG-wine23}.

The research agenda of designing online algorithms against the online optimum opens avenues for online stochastic optimization, and fits into the trend of ``beyond worst-case analysis''. Besides the obvious fact that the worst-case competitive ratio lower bounds are surpassed using the new benchmark, the framework has also witnessed new algorithmic ideas and useful technical tools in online optimization.

\subsection{Our Contributions}
In this work, we adapt the order-competitive analysis framework and study the online stochastic matching problem with \emph{unknown arrival order}, 


Our main result confirms that the $0.5$ barrier is beatable against the online optimum, even without the prior knowledge of the arrival order.

\begin{Theorem}
\label{thm:main}
There exists a poly-time $0.5+\Omega(1)$ order-competitive algorithm for the online stochastic matching problem.	
\end{Theorem}

Our result is aligned with the traditional online algorithm literatures, as the main difficulty comes from the lack of information, i.e., the arrival order. It complements the existing approximation results that focus on the computational perspective of the online stochastic matching problem with known arrival order.

The exact ratio of our order-unaware algorithm, provided in Section~\ref{sec:ratio}, is significantly smaller than the state-of-the-art approximation ratio by Braverman et al.~\cite{BKPS-arxiv24}.
On the other hand, due to the order-unawareness of our algorithm, our result implies a $0.5+\Omega(1)$ approximation algorithm for \emph{stochastic} arrival orders. To the best of our knowledge, previous approximation algorithms only work for deterministic arrivals.

\paragraph{Techniques.}

Observe that the online optimum is defined with respect to a fixed arrival order of the online vertices. All previous approximation algorithms rely on the same LP relaxation which utilizes the knowledge of the arrival order, hence, are not applicable in our setting.

To beat the $0.5$ barrier, we first examine the performance of the following class of algorithms: upon the arrival of an online vertex $t$, the online vertex proposes to offline vertex $i$ with probability $x_{it}/p_t$. Then, from the perspective of each offline vertex $i$, this can be viewed as a single-choice prophet inequality problem. We then examine whether the standard fixed-threshold algorithm, which is known to be $0.5$ competitive against offline optimum, already achieves a $0.5 + \Omega(1)$ competitive ratio. 
Furthermore, it is without loss of generality to assume that the online optimum is sufficiently close to the offline optimum. Indeed, if we can design a $0.5+\Omega(1)$ order-competitive algorithm under the assumption, a randomization between the classic $0.5$ competitive algorithm and the new algorithm would then guarantee a $0.5+\Omega(1)$ order-competitive algorithm for every instance.
Our approach heavily exploits the above two observations. 

\begin{itemize}
    \item We first introduce a family of non-adaptive algorithms whose behaviors are independent of the arrival order. These are natural generalizations of fixed-threshold algorithms in the single-choice setting. We prove that if none of the non-adaptive algorithms can guarantee a $0.5+\Omega(1)$ competitive ratio, the instance must be shaped in a so-called ``free-deterministic'' structure. This structure is motivated from the classic hard instance for single-choice setting (a.k.a., the prophet inequality): The first value $1$ arrives with probability $1$, while the second value $1/\epsilon$ arrives with probability $\epsilon$. As this instance suggests that no algorithm for single-choice setting can be $0.5 + \Omega(1)$ competitive, we provide a reversed lemma for the hard instance: For a single-choice instance, if no fixed-threshold algorithm can be $0.5 + \Omega(1)$ competitive, the instance can be decomposed into two parts: the free part and the deterministic part. Both parts contribute half of the expected reward while the total arrival probability of the free part is sufficiently close to $0$. 
    
    By further applying the ``free-deterministic'' structure to the matching setting, we may partition the edges into two parts: the free edges and the deterministic edges. Both parts contribute half of the expected reward, while the matching probability of the free edges is  close to $0$. The formal statements are provided in Section~\ref{sec:decomp}. We believe the lemmas shall find further applications in beating the worst-case competitive ratio against the online optimum. E.g., for designing $0.5+\Omega(1)$ order-competitive algorithm (or even approximation algorithm) for matroid online Bayesian selection~\cite{DP-sagt24}.
    

    
    \item Furthermore, we show that each online vertex has at most one free edge. En route, we introduce a novel linear program that we name as the slackness LP. Informally, the LP relaxation examines if there are two approximately optimal, but significantly different offline matching solutions. If so, we prove that there must exist a $0.5+\Omega(1)$ competitive algorithm against the offline optimum.

    A technical highlight is that approximating the diameter (e.g., L1-norm) of a polytope is generally computationally hard. Our relaxation relies on the ``free-deterministic'' decomposition from the first step.
    The formal statement is provided in Section~\ref{sec:large_slackness}.
    
    \item Finally, we focus on those restricted instances that satisfy the properties as derived above. We provide an $\nicefrac{5}{8}-o(1)$ order-competitive algorithm when the online optimum is $o(1)$-close to the offline optimum.
    
    Our new order-unaware algorithm puts every offline vertice into two stages. At the first stage of an offline vertex, it keeps collecting the corresponding free edges until it collects (almost) all these free rewards, at which moment the offline vertex switches to the second stage and commits to a (fractional) matching of the deterministic edges immediately. 
    
    According to the property of the free edges, our algorithm collects in expectation close to half of the offline matching from the first stage, while each offline vertex is only matched with probability close to $0$ at the end of its first stage.
    We then utilize the assumption that the online optimum approximately equals to the offline optimum and deduce that the optimal order-aware algorithm must make the same matching decisions as ours at the first stage. As a consequence, half of the expected matching produced by the optimal order-aware algorithm is from the second stage deterministic edges. The corresponding second-stage matching problem corresponds to an online edge-weighted matching setting, with the role of online and offline vertices being swapped. 

    A warm-up version of our order-unaware algorithm is given in Section~\ref{sec:warmup} with extra assumptions on the input.
    The formal algorithm and analysis are provided in Section~\ref{sec:small_slackness}.
    An important remark is that even under the assumption that the online optimum equals the offline optimum, no order-unaware algorithm can achieve an order-competitive ratio better than $1-\omega(1)$ (see Observation~\ref{obs:negative}). This is in contrast to the single-choice Bayesian selection problem according to the result of Chen, Sun, and Tang~\cite{CST-arxiv24}. We interpret this observation as an extra difficulty of online matching, in comparison to the online single-choice setting, complementing the computational hardness result of Papadimitriou et al.~\cite{PPSW-MOR24}.
\end{itemize}

\subsection{Further Related Works}

Besides the above mentioned papers, there are more recent works on studying online optimum. 
In the single-choice setting, Niazadeh, Saberi, and Shameli~\cite{NSS-wine18} studied fixed-threshold algorithms; Ezra, Feldman, and Tang~\cite{EFT-arxiv24} studied anonymous (blind-to-identity) algorithms.
As a generalization of online stochastic matching, Braun, Kesselheim, Pollner, and Saberi~\cite{BKPS-arxiv24} designed a $0.5+\Omega(1)$ approximation algorithm for online capacitated resource allocation, where each online vertex can be matched to multiple offline vertices. 

The $0.5$ impossibility result of online stochastic matching against the offline optimum is inherited from the single-choice setting, a.k.a., the prophet inequality~\cite{krengel1977semiamarts, krengel1978,samuel1984}. 
Besides comparing against the online optimum, researches have studied multiple variants of the single-choice setting for beyond worst-case results, including i.i.d. arrivals~\cite{10.1214/aop/1176993861,DBLP:conf/stoc/AbolhassaniEEHK17,DBLP:journals/mor/CorreaFHOV21, DBLP:journals/corr/abs-2408-07616}, random arrivals~\cite{DBLP:journals/siamdm/EsfandiariHLM17,DBLP:journals/mor/CorreaFHOV21,DBLP:conf/soda/EhsaniHKS18,DBLP:conf/sigecom/AzarCK18,DBLP:journals/mp/CorreaSZ21,H-arxiv23}, order selection settings~\cite{DBLP:conf/stoc/ChawlaHMS10, DBLP:journals/ior/BeyhaghiGLPS21,DBLP:conf/sigecom/0001SZ20, PT-focs22, BC-ec23}, and with cancellation costs~\cite{DBLP:conf/stoc/EkbataniNNV24}. Beyond single-choice setting, the prophet inequality model has been extensively studied for multiple classes of combinatorial value functions, including submodular/XOS functions \cite{FGL-SODA15, DFKL-FOCS17} and subaddtive functions \cite{DBLP:conf/focs/DuttingKL20, DBLP:conf/stoc/CorreaC23}. All these function classes include the matching setting (unit-demand function) as a special case. 

The study of online bipartite matching dates back to the seminal work of Karp, Vazirani, and Vazirani~\cite{KVV-STOC90}. A number of variants have been studied in the literature. Below, we discuss the most related works. Refer to a recent survey by Huang, Tang, and Wajc~\cite{HTW-arxiv24} for a more detailed discussion.
A restricted version of our online stochastic matching is to assume i.i.d. arrivals. This problem was first introduced by Feldman, Mehta, Mirrokni, and Muthukrishnan~\cite{FMMM-FOCS09} for unweighted graphs. This line of research~\cite{JL-mor14,BGMS-nips21,HS-stoc21,HSY-stoc22,TWW-STOC22,MGS-mor12,Yan-soda24,QFZW-wine23} has been focused on beating the $1-\nicefrac{1}{e}$ barrier from the adversarial setting. 
An important component of our algorithm is the edge-weighted online bipartite matching with free disposal. The problem was introduced by Feldman et al.~\cite{FKMMP-wine09}, who designed an optimal $1-\nicefrac{1}{e}$ competitive fractional algorithm. The breakthrough result by Fahrbach, Huang, Tao, and Zadimoghaddam~\cite{FHTZ-jacm22} designed a randomized integral algorithm,  beating the trivial $0.5$ ratio for the first time. The ratio is further improved by \cite{GHHNYZ-focs21,BC-focs21}.

%% file: prelim.tex
\section{Preliminaries}



An instance $I$ of our problem consists of an edge-weighted bipartite graph $G=(U\cup V,E,w)$, a probability vector $p \in [0,1]^V$, and an arrival order of the online vertices $V$.
We use $i \in U=[n]$ to denote an offline vertex and $t \in V=[T]$ to denote an online vertex. 

At each step $t$, our algorithm observes the identity of the arriving vertex $t$. The vertex is realized independently with probability $p_t$.  The algorithm decides immediately whether and how to match $t$ if it is realized. If an online vertex $t$ is matched to an offline vertex $i$, reward $w_{it} \geq 0$ is collected. Each offline vertex $i$ can be matched to at most one online vertex $t$.

We focus on designing order-unaware algorithms that only know the graph (i.e., edge weights $\{w_{it}\}_{i \in U, t \in V}$) and the probabilities $\{p_t\}_{t \in V}$ in advance, but not the arrival order. Our goal is to compete against the optimal order-aware online algorithm that also knows the arrival order in advance. We abuse $\philo$ to denote both the optimal online algorithm that knows the order and its expected performance.
More formally, the order-competitive ratio of an order-unaware algorithm with $\alg$ being its expected performance is defined as
\[
\min_{I} \frac{\alg(I)}{\philo(I)}~,
\]
In contrast, the traditional competitive analysis studies competitive ratio, that sets the expected optimal offline matching as the benchmark / denominator. 

\smallskip
\subsection{Ex-ante optimum and Online optimum}
Consider the ex-ante optimum as an upper bound of the offline optimum matching.
\begin{align*}
\label{eqn:lpexante}
\max_{x}: \quad & \sum_{i,t} \wit \cdot \xit \tag{$\lpexante$} \\
\text{subject to :} \quad & \sum_i \xit \le p_t & \forall t \in R \\
& \sum_t \xit \le 1	& \forall i \in L
\end{align*}
We use $x^*$ to denote the optimal solution of the ex-ante relaxation. For any $x$, let $\lp_i(x) \eqdef \sum_t \wit \cdot \xit$ and $\lp(x) \eqdef \sum_i \lp_i(x)$.
We use $\feasib$ to denote the set of feasible solutions of the ex-ante relaxation.

For any $x \in \feasib$, we define $q_i(x) \eqdef \sum_{t \in V} x_{it}$ and $q_t(x) \eqdef \sum_{i \in U} x_{it}$, i.e., notation $q_i(x)$ and $q_t(x)$ represent the total probability spent by an offline vertex $i$ and an online vertex $t$ respectively. A solution $x$ is feasible ($x \in \feasib$) iff $q_i(x) \leq 1$ and $q_t(x) \leq p_t$ are simultaneously satisfied for every $i \in U$ and $t \in V$.

We use $\yit^*$ to denote the probability that edge $(i,t)$ is matched by the optimal online algorithm $\philo$. Then $y^*$ naturally satisfies the constraints of the ex-ante optimum. Hence,
\[
\philo = \lp(y^*) \le \lp(x^*) = \lpexante
\]
Existing works~\cite{SW-icalp21,BDL-EC22,NSW-arxiv23,BDPSW-arxiv24} all make use of the following constraints on $y^*$:
\begin{equation}
\yit^* \le \left(1-\sum_{s<t} \yit[s]^* \right) \cdot p_t~, \quad \forall (i,t) \in E \label{eq:onlinerelax}
\end{equation}
since the realization of vertex $t$ is independent of previous matching decisions. We note that the arrival order of online vertices is not provided to our algorithm, and thus we cannot use \eqref{eq:onlinerelax} in advance. 
Nevertheless, we shall also use \eqref{eq:onlinerelax} to bound the performance of $\philo$ in Section~\ref{sec:small_slackness}.

\subsection{Baseline algorithms}
Feldman, Gravin, and Lucier~\cite{FGL-SODA15}, and Ezra, Feldman, Gravin, and Tang~\cite{EFGT-MOR22} designed two different $0.5$-competitive algorithms, even against the ex-ante optimum. The first algorithm is threshold-based using static prices, and the second is probability-based using techniques from online contention resolution schemes.

For our purpose, we consider the following baseline algorithm, that can be viewed as  a mixture of the existing two algorithms. Specifically, our baseline algorithm is probability-based for online vertices; and is threshold-based for offline vertices. 

Given an arbitrary $x \in \feasib$, let $\algbase(x)$ be the following baseline algorithm:
\begin{itemize}
\item Upon the arrival of an online vertex $t$, it non-adaptively proposes to one of its neighbors with probability $\nicefrac{\xit}{p_t}$. The first family of constraints guarantees the validity of this step, i.e., $\sum_{i} \nicefrac{\xit}{p_t} \le 1$.
\item From the perspective of each offline vertex $i$, it receives a proposal from $t$ with probability $\xit$ independently. This can be interpreted as a single-choice prophet problem and we consider a single-threshold algorithm that accepts the first proposed edge whose value is at least $\tau_i(x)$. Here, $\tau_i(x)$ is defined as the maximizer of the following quantity:
\[
\stlb_i(x) \eqdef \max_{\tau} \sum_{t} \prod_{s: \tau \le \wit[s] < \wit} (1-\xit[s]) \cdot \xit \cdot \wit~.
\]
\end{itemize}
The quantity $\stlb_i(x)$ is a lower bound of the expected weight we collect from vertex $i$. Indeed, it is the value when the edges (or their corresponding vertices) arrive in ascending order of their weights.
Moreover, the classic prophet inequality tells us that the above value is at least half of the ex-ante optimum, i.e., $\lp_i(x)$.
Let $\algbi(x)$ be the expected weight we collect from vertex $i$ using $\algbase(x)$. We summarize the above discussion and omit the folklore proof.
\begin{Lemma}
\label{lem:algbase}
$\algbi(x) \ge \stlb_i(x) \ge 0.5 \cdot \lp_i(x)~.$
\end{Lemma}

A natural choice of the parameter is to use $x^*$, the optimal ex-ante solution. As a consequence of Lemma~\ref{lem:algbase}, it is at least $0.5$-competitive against the ex-ante optimum. 
\begin{Corollary}
$\algbase(x^*) \ge 0.5 \cdot \lpexante$.
\end{Corollary}

%% file: warm_up.tex
\section{Warm up: the deterministic and free vertices setting}
\label{sec:warmup}

As a warm up and to convey our main algorithmic idea, we focus on a special yet non-trivial family of instances with the following two extra assumptions:
\begin{enumerate}
\item Every online vertex is either \emph{deterministic} with $p_t = 1$ or \emph{free} with $p_t = 0^+ = o(1/(nT))$. We use $\dt$ to denote the set of deterministic vertices and $\fr$ to denote the set of free vertices. 
\item The graph is online-vertex-weighted, i.e., $\wit = w_t$ for every $(i,t) \in E$. Without loss of generality, we assume each online vertex has at least one neighbor, as otherwise we can simply remove the vertex from the instance. Furthermore, let $v_t \eqdef w_t \cdot p_t$ be the expected value of vertex $t$. We note that the magnitude of $v_t$ is roughly the same for deterministic and free vertices, i.e., we have $w_{t} \gg w_{t'}$ for a free vertex $t$ and a deterministic vertex $t'$. 
\end{enumerate}
This family of instances play an important role in the literature of prophet inequalities, i.e., when there is only one offline vertex. Indeed, both the optimal $\nicefrac{1}{2} = 0.5$ competitive ratio upper bound~\cite{krengel1977semiamarts, krengel1978,samuel1984} and the optimal $\nicefrac{1}{\varphi} \approx 0.618$ order-competitive ratio upper bound~\cite{EFGT-SODA23} for deterministic algorithms are established by this family of instances.

\paragraph{Disclaimer.}
We are not going to present any formal proofs within this section, although we have already made strong assumptions on the input instance.
Our plan is to sketch an order-unaware algorithm with order-competitive ratio strictly greater than $0.5$, and highlight the main difficulty of our analysis, so that the formal (and mathematically heavy) analysis provided in later sections become easier to digest for our readers. 
We shall make hand-waving arguments that are \emph{with loss of generality}.
Indeed, naively formalizing the analysis below can only lead to an improved order-competitive ratio of $0.5 + o(1)$. 
To get a constant improvement, we have to include significantly many technical details that might distract readers from our algorithmic idea.
To this end, we hide these details and make statements in an informal but hopefully intuitive way.

\paragraph{Ex-ante optimum.} 
Let $M^*$ be the maximum matching between $U$ and $\dt$. 
For each $i \in U$, we use $i^*$ to denote the neighbor of $i$ in $M^*$. Let $w_i \eqdef w_{i,i^*}$ be the weight of the edge $(i,i^*)$. If $i$ is unmatched in $M^*$, let $i^* = \perp$ and $w_i = 0$.

Furthermore, since vertices in $\fr$ are realized with negligible probability that can be taken for free, the expected offline maximum matching equals 
\[
\lpexante =\sum_{i \in U} w_i + \sum_{t \in \fr} v_t~.
\] 


\paragraph{Baseline algorithm.} 
We consider the following baseline algorithm that is captured by a function $x: \fr \to U$, where $x(t)$ is a neighbor of $t$. 
For each offline vertex $i \in U$, let $\fr_i(x) \eqdef \{t \in \fr: x(t) = i \}$ and $v_i(x) \eqdef \sum_{s \in \fr_i(x)} v_s$. Each offline vertex $i$ weighs the following two options:
\begin{itemize}
\item waits until the arrival of $i^* \in \dt$, and collects $w_i$;
\item greedily accepts any free neighbors from $\fr_i(x)$ if realized, and collects $v_i(x)$. Recall that each free vertex $s$ is realized with only negligible probability but contributes $v_s$ in expectation.
\end{itemize}
The better of the options provide an expected reward of $\max \left( w_i, v_i(x) \right)$.

This baseline algorithm is order-unaware and produces a matching with expected weight
\begin{equation}
\label{eqn:warmup_base}
\algbase(x) = \sum_{i \in U} \max \left( w_i, v_i(x) \right) \ge \frac{1}{2} \cdot \sum_{i \in U} \left( w_i + v_i(x) \right) = \frac{1}{2} \cdot \lpexante~,
\end{equation}
where the inequality achieves equality only when $w_i = v_i(x)$ for every $i$.

\paragraph{Online optimum.} As a consequence of equation~\eqref{eqn:warmup_base}, if the online optimum is strictly less than the offline optimum, the baseline algorithm then has an order-competitive ratio strictly greater than $0.5$. Thus, we can restrict ourselves to instances with the following assumption.
\begin{enumerate}\addtocounter{enumi}{2}
\item $\philo$ equals to $\lpexante~$  asymptotically, i.e., $\philo \to \lpexante$ when $p_t \to 0$ for all free vertices.
\end{enumerate}

\paragraph{Balancedness and Uniqueness.}
 Furthermore, equation \eqref{eqn:warmup_base} asserts that an arbitrary choice of the mapping $x$ gives an at least $0.5$ competitive ratio against the optimal ex-ante benchmark. 
Thus, if the baseline algorithm fails to achieve a competitive ratio strictly greater than $0.5$ (i.e., \eqref{eqn:warmup_base} becomes an equality), the instance must satisfy the following two conditions:
\begin{enumerate}\addtocounter{enumi}{3}
\item (Balancedness) The contribution of each offline vertex to the ex-ante optimum can then be separated to the deterministic part $w_i$ and the free part $v_i(x)$, and they must be balanced, i.e., $w_i = v_i(x)${, since otherwise $\max(w_i, v_i(x))$ is strictly greater than the average value $\nicefrac{1}{2} \cdot (w_i+v_i(x))$.} Consequently, we have that $\sum_{i \in U} v_i(x) = \nicefrac{1}{2} \cdot \lpexante$;
\item (Uniqueness) Each free vertex has a unique neighbor. Otherwise consider two mappings $x, x'$ whose only difference is that $x(t) \ne x'(t)$ for some $t$. {Then, the two solutions $x,x'$ cannot meet the balancedness assumption at the same time for vertices $x(t)$ and $x'(t)$.} As a consequence, it cannot happen that $\algbase(x)=\algbase(x')=\nicefrac{1}{2} \cdot \lpexante$. In other words, there exists a unique valid mapping $x$.
\end{enumerate}

\paragraph{Order-unaware algorithm.}
So far, our main character of the problem, the arrival order, has not been involved. 
We now design an order-unaware algorithm that produces a matching with expected weight $0.75 \cdot \lpexante$, with all the above 5 assumptions. 

\begin{tcolorbox}[title=Order-unaware algorithm under the 5 assumptions]
\label{alg:warmup}
Initialization:
\begin{enumerate}
\item Let $x(t)$ be the unique neighbor of $t \in \fr$ and define $x(t) = \perp$ for $t \in \dt$.
\item Let all offline vertices be in stage $1$.
\end{enumerate}
Upon the arrival of $t$:
\begin{enumerate}[Step 1:]
\item Match $t$ to $x(t)$ if $t$ is realized and $x(t)$ is unmatched.
\item If $t$ is the last free neighbor of vertex $i$, move vertex $i$ to stage $2$, set $t_i = t$, and update $x(s) = i$ for the largest weight deterministic unassigned neighbor $s$ of $i$, i.e.,
\[
s = \argmax_{s > t} \{ w_s: x(s) = \perp, s \in \dt_i \}~.
\]
\end{enumerate}
\end{tcolorbox}
Observe that, each offline vertex $i$ collects an expected reward of $v_i(x)$ in stage $1$ and then switchs to stage $2$. 
According to the balancedness assumption, this guarantees an expected reward of $\nicefrac{1}{2} \cdot \lpexante$ in total and we are left to show that the algorithm collects at least $\nicefrac{1}{4} \cdot \lpexante$ from stage $2$. This can be derived through the following two steps:
\begin{itemize}
\item According to the assumption 5 (the uniqueness property), if an offline vertex $i$ is matched to a deterministic vertex before free vertex $t$ that satisfies $x(t) = i$ arrives, the online algorithm fails to collect $v_t$ compared to the offline optimum. Since assumption 3 suggests that the online optimum equals to offline optimum, this case cannot happen. Therefore, our matching decisions in stage $1$ are optimal. That is, the optimal order-aware online algorithm also ignores all deterministic neighbors of $i$ when $i$ is in stage $1$. In other words, half of the online optimum comes from second-stage deterministic edges:
\[
E_{(2)} \eqdef \{ (i,t) \in \dt: t > t_i\}~.
\]
Consequently, we have that 
\begin{equation}
\label{eqn:warmup_second_opt}
\mwm\left(E_{(2)}\right) = \nicefrac{1}{2} \cdot \lpexante~,
\end{equation}
where $\mwm(F)$ is the maximum weight matching on the edge set $F$.
\item Finally, consider an online bipartite matching algorithm as the following:
\begin{itemize}
\item Let the underlying graph be $(U,V,E_{(2)})$, and each $t \in V$ is associated with a weight $v_t$.
\item Vertices $V$ are given upfront, while vertices $i \in U$ arrive in a sequence in ascending order of $t_i$.
\item Upon the arrival of $i$, its incident edges are revealed and we make matching decisions between $i$ and $V$.
\end{itemize}
Notice that in this problem, $U$ and $V$ switch the role of online and offline vertices as in the original problem; and this is a standard online (offline-vertex-weighted) bipartite matching problem.
The step 2 of our algorithm can then be interpreted as a straightforward greedy algorithm for this problem, that is $\nicefrac{1}{2}$-competitive with respect to the benchmark $\mwm(E_{(2)})$. Together with \eqref{eqn:warmup_second_opt}, we conclude the analysis.
\end{itemize}

\paragraph{Remark.}
The step 2 of our order-unaware algorithm can be substituted by a better vertex-weighted online bipartite matching algorithm, e.g., the Ranking algorithm~\cite{AGKM-SODA11}. This would result in a $1-\nicefrac{1}{2e}$ order-competitive algorithm under the same $5$ assumptions above. Nevertheless, any constant competitive algorithm here serves the purpose for beating the $0.5$ order-competitive ratio. Moreover, we construct an instance below showing that any order-unaware algorithm cannot achieve the optimal performance guarantee, even under the 5 assumptions above. 
This observation contrasts with the single-choice setting, where \cite{CST-arxiv24} established a $1$ order-competitive ratio under only assumption 3. 


\begin{Observation}
\label{obs:negative}
No order-unaware algorithm can achieve an order-competitive strictly better than $\frac{11}{12}$, even for instances satisfying the 5 assumptions above.
\end{Observation}



\begin{proof}
Consider the example as shown in the figure below. There are three offline vertices, labeled as $1, 2, 3$. For the online vertices, there are $3$ free vertices $F_1, F_2, F_3$ and three deterministic vertices $D_{12}, D_{13}, D_{23}$. All vertices have the same vertex weight.

The arrival order of online vertices is stochastic. We consider the following two arrival orders, each with probability $0.5$:
\begin{itemize}
    \item Order 1: $F_1, F_2, D_{12}, D_{13}, F_3, D_{23}$.
    \item Order 2: $F_1, F_2, D_{12}, D_{23}, F_3, D_{13}$
\end{itemize}

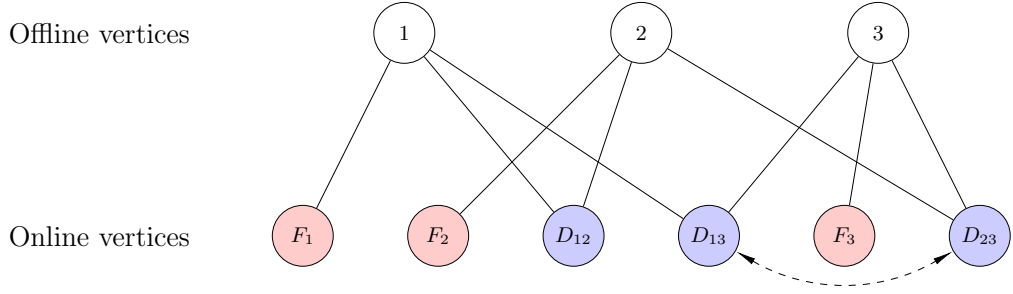
\begin{figure}[h]
    \centering
    
    \scalebox{0.9}{
\begin{tikzpicture}
    \footnotesize
    \tikzset{
        vertex/.style={circle, draw, minimum size=0.9cm, align = center},
        edge/.style={-}
    }

    \node[vertex] (A1) at (1.5,4) {$1$};
    \node[vertex] (A2) at (5,4) {$2$};
    \node[vertex] (A3) at (8.5,4) {$3$};

    \node[vertex, fill=red!20] (B1) at (0,1) {$F_1$};
    \node[vertex, fill=red!20] (B2) at (2,1) {$F_2$};
    \node[vertex, fill=red!20] (B3) at (8,1) {$F_3$};
    \node[vertex, fill=blue!20] (B4) at (4,1) {$D_{12}$};
    \node[vertex, fill=blue!20] (B5) at (6,1) {$D_{13}$};
    \node[vertex, fill=blue!20] (B6) at (10,1) {$D_{23}$};
    \node at (-3, 1) {\large{Online vertices}};
    \node at (-3, 4) {\large{Offline vertices}};
    \draw[edge] (A1) -- (B1);
    \draw[edge] (A2) -- (B2);
    \draw[edge] (A3) -- (B3);
    \draw[edge] (A1) -- (B4);
    \draw[edge] (A1) -- (B5);
    \draw[edge] (A2) -- (B4);
    \draw[edge] (A2) -- (B6);
    \draw[edge] (A3) -- (B5);
    \draw[edge] (A3) -- (B6);
    \draw[<->, dashed, bend right, >={Latex[length=3mm, width=1mm]}] (B5) to (B6);
\end{tikzpicture}
    }   
    \caption{Hard Instance}
    \label{fig:notenough}
\end{figure}

It is easy to verify that the input graph satisfies all 5 assumptions listed in this section. Now, we show that no order-unaware algorithm can achieve an order-competitive ratio strictly better than $\nicefrac{11}{12}$ against online optimum (note that the online optimum equals to offline optimum for this instance).

Consider all possible actions for the first three arriving online vertices $F_1, F_2, D_{12}$. The algorithm must take the value of $F_1$ and $F_2$, otherwise the optimal competitive ratio is at most $\nicefrac{5}{6}$. For $D_{12}$, there can be three possible actions:
\begin{itemize}
    \item Algorithm matches $D_{12}$ to offline vertex $1$. In this case, the algorithm loses at least $\frac{1}{6}$ of the online optimum when the arrival order is Order 1, because either it matches $D_{13}$ to offline vertex $3$ and misses $F_3$, or it misses $D_{13}$. Therefore, the competitive ratio is at most $\frac{11}{12}$.
    \item Algorithm matches $D_{12}$ to offline vertex $2$. In this case, the algorithm loses at least $\frac{1}{6}$ of the online optimum when the arrival order is Order 2. The proof is symmetric to the first case.
    \item Algorithm skips $D_{12}$. In this case, the competitive ratio is at most $\frac{5}{6}$. 
\end{itemize}
\end{proof}

%% file: preprocessing.tex
\section{Free-deterministic Decomposition and Pruning}
\label{sec:decomp}

Starting from this section, we generalize the algorithm in Section~\ref{sec:warmup} by removing those extra assumptions we made in Section~\ref{sec:warmup}.
Those assumptions might look quite strong at the first glance. But recall that the baseline algorithm is at least $0.5$-competitive against the ex-ante optimum, we only need to worry about those instances when the baseline algorithm is no better than $(0.5+\eps)$-competitive. Intuitively, this would lead to quite restrictive instances. 

Our first step is to formalize the free-deterministic structure, that corresponds to assumptions 1 and 4 in  Section~\ref{sec:warmup}.
The following lemma studies single-choice prophet inequalities when no fixed-threshold algorithm can guarantee a competitive ratio strictly better than $0.5$.
We believe the lemma can find further applications in online stochastic optimization against the online optimum.


\begin{Lemma}
\label{lem:free_prob}
For every $\mu<0.5, \beta>0.5+\mu$, if $\stlb_i(x) < (0.5+\mu) \cdot \lp_i(x)$, where we recall
\[
\stlb_i(x) \eqdef \max_{\tau} \sum_{t} \prod_{s: \tau \le \wit[s] < \wit} (1-\xit[s]) \cdot \xit \cdot \wit~
\]
represents the lower bound of the expected weight we can collect from offline vertex $i$ via a fixed-threshold algorithm, we have 
\begin{flalign*}
1) & \sum_{t} \xit \cdot \one [w_{it} > \beta \cdot \lp_i(x)] \le \delta = \sqrt{\frac{4\mu}{\beta - 0.5 - \mu}}~; & \\
2) & \sum_{t} \xit \cdot \wit \cdot \one [w_{it} > \beta \cdot \lp_i(x)] \le \frac{0.5+\mu}{1-\delta}\cdot \lp_i(x)~. &
\end{flalign*}
\end{Lemma}

\begin{proof}
For notation simplicity, within the proof of this lemma, we normalize the value of $\lp_i(x)$ to be $1$ and enumerate the edges in ascending order of their edges, i.e., we reorder the online vertices so that $w_{i1} \le w_{i2} \le \cdots \le w_{in}$. 
Let $k,l$ be the indices that
\[
w_{i,k-1} < \left(0.5+\mu \right) \le w_{i,k} \quad \text{and} \quad w_{i,l-1} < \beta \le w_{il}
\]
Then $\sum_{t} \xit \cdot \one [w_{it} > \beta]  = \sum_{t=l}^{n} \xit$, and we denote this value by $\delta$.
Since $\lp_i(x)=1$, we have that
\begin{equation}
\label{eqn:sumequalone}
\sum_{t=k}^{n} \xit \cdot \wit =1-\sum_{t=1}^{k-1} \xit \cdot \wit \ge 1 - \left(0.5+\mu \right) \cdot \left( 1-\sum_{t=k}^{n} \xit \right),	
\end{equation}
where the inequality holds since $\wit \le 0.5+\mu$ for every $t \le k-1$ and $\sum_{i=1}^{n} \xit \le 1$.

For the induced single-choice prophet problem for vertex $i$, consider using the single-threshold algorithm with $(0.5+\mu)$.
According to the definition of $\stlb_i$, we have that
\begin{multline*}
(0.5+\mu) > \stlb_i(x)\ge  \sum_{t=k}^{n} \prod_{s=k}^{t-1} (1-x_{is}) \cdot \xit \cdot \wit \\
= \sum_{t=k}^{n} \left[\prod_{s=k}^{t-1} (1-x_{is}) - \prod_{s=k}^{n-1} (1-x_{is}) \right]\cdot \xit \cdot \wit + \prod_{s=k}^{n-1} (1-x_{is}) \cdot \sum_{t=k}^n \xit \cdot \wit \\
\ge \sum_{t=k}^{n} \left[\prod_{s=k}^{t-1} (1-x_{is}) - \prod_{s=k}^{n-1} (1-x_{is}) \right] \cdot \xit \cdot (0.5+\mu) + \sum_{t=l}^{n} \left[\prod_{s=k}^{t-1} (1-x_{is}) - \prod_{s=k}^{n-1} (1-x_{is}) \right] \cdot \xit \cdot (\beta - 0.5 - \mu)\\
+ \prod_{s=k}^{n-1} (1-x_{is}) \cdot \left[ 1 - \left(0.5+\mu\right) \cdot \left( 1-\sum_{t=k}^{n} \xit \right) \right] \\
= \sum_{t=k}^{n} \prod_{s=k}^{t-1} (1-x_{is}) \cdot \xit \cdot (0.5+\mu) + \sum_{t=l}^{n} \left[\prod_{s=k}^{t-1} (1-x_{is}) - \prod_{s=k}^{n-1} (1-x_{is}) \right] \cdot \xit \cdot (\beta -0.5 - \mu)\\
+ \prod_{s=k}^{n-1} (1-x_{is}) \cdot \left( 1 - \left(0.5+\mu\right) \right)~.
\end{multline*}
Here, the last inequality follows from the definitions of $k$ and $l$, and inequality \eqref{eqn:sumequalone}. Notice that for the first term of the right hand side, we have $\sum_{t=k}^{n} \prod_{s=k}^{t-1} (1-x_{is}) \cdot \xit = 1-\prod_{t=k}^n (1-\xit)$. Thus, the inequality simplifies as the following:
\begin{align*}
(0.5 + \mu) & > (0.5+\mu) \left(1 - \prod_{t=k}^{n} (1-\xit) \right) + \prod_{t=k}^{n-1} (1-\xit) \cdot \left( 0.5 - \mu \right) \\
& + \prod_{t=k}^{l-1}(1-\xit) \cdot \left( 1 - \prod_{t=l}^{n-1}(1-\xit) - \prod_{t=l}^{n-1} (1-\xit) \cdot \sum_{t=l}^{n} \xit \right) \cdot (\beta -0.5 - \mu).
\end{align*}
Dividing the two sides by $\prod_{t=k}^n (1-\xit)$, we have
\begin{align*}
0.5 + \mu & \ge \frac{0.5-\mu}{1-\xit[n]} +(\beta - 0.5 - \mu) \cdot \left[ \prod_{t=l}^{n}(1-\xit)^{-1} - 1 - \frac{\sum_{t=l}^n \xit}{1-\xit[n]} \right] \\
& \ge 0.5-\mu +(\beta - 0.5 - \mu) \cdot \left[ \prod_{t=l}^{n}(1-\xit)^{-1} - 1 - \sum_{t=l}^n \xit \right] \\
& \ge 0.5-\mu +(\beta - 0.5 - \mu) \cdot \left( e^\delta - 1 - \delta \right) \\
& \ge 0.5-\mu +(\beta - 0.5 - \mu) \cdot \frac{1}{2} \cdot \delta^2
\end{align*}
This concludes that $\delta \le \sqrt{\frac{4\mu}{\beta-0.5-\mu}}$.

Finally, consider using threshold $\tau = \beta \cdot \lp_i(x)$ for the induced single-choice prophet problem for vertex $i$. According to the definition of $\stlb_i$, we have that
\[
0.5 + \mu > \stlb_i(x) \ge \sum_{t=l}^{n} \prod_{s=l}^{n-1} (1-x_{is}) \cdot x_{it} \cdot \wit \ge \left(1-\sum_{t=l}^{n} \xit \right) \cdot \sum_{t=l}^{n} \xit \cdot \wit \ge (1-\delta) \cdot \sum_{t=l}^{n} \xit \cdot \wit.
\]
\end{proof}





Intuitively, \Cref{lem:free_prob} suggests that for any single-item Prophet Inequality instance where the optimal fixed-threshold algorithm achieves a $(0.5 + o(1))$-competitive ratio, we can decompose the instance into high-value and low-value parts, with each part contributing approximately $0.5$ of the LP value, while the high-value part consumes only an $o(1)$ fraction of the probability. We then generalize \Cref{lem:free_prob} to the Prophet matching setting, showing that  any feasible solution $a \in \feasib$ with $\stlb(a) < (0.5+\gamma) \lp(a)$ can be pruned so that the resulting instance consists of a ``free'' part that contributes almost half of the total weight while only consumes negligible matching probability of each offline vertex. 

\begin{Lemma}[Decomposition Lemma]
\label{lma:frdt}
For every $\gamma<10^{-4},\alpha \ge 1$, and $a \in \feasib$ such that $\stlb(a) \leq (0.5 + \gamma) \cdot \lp(a)$, there exists $\tilde{a}$ such that
\begin{itemize}
    \item $\forall (i,t) \in E, \tilde a_{it} \le a_{it}$~;
    \item $\lp(\tilde{a}) \ge (1-\gamma^{\nicefrac{1}{4}}) \cdot \lp(a)$.
\end{itemize}
Furthermore, for every $i \in U$, let $\tilde a^L_{it} = \tilde a_{it} \cdot \one [w_{it} > \alpha \cdot \lp_i(\tilde a)]$. We have 
\begin{itemize}
    \item $\sum_{t} \tilde{a}^{L}_{it} \le \gamma^{\nicefrac{1}{4}}$~;
    \item $(0.5 - (\alpha+2) \gamma^{\nicefrac{1}{4}}) \cdot \lp_i(\tilde{a}) \leq \lp_i(\tilde a^L) \leq (0.5 + \gamma^{\nicefrac{1}{4}} ) \cdot \lp_i(\tilde a)$~.
\end{itemize}
\end{Lemma}

\begin{proof}
Let $U_0 = \{ i: \stlb_i(a) < (0.5+\gamma^{\nicefrac{3}{4}}) \cdot \lp_i(a)\}$. 
We construct $\tilde{a}$ as the following by pruning $a$:
\[
\tilde{a}_{it} = 
\begin{cases}
a_{it} & i \in U_0 \\
0 & i \notin U_0	
\end{cases}
\]
We start with the first part of the statement. Straightforwardly, we have $\tilde a_{it} \le a_{it}$. Moreover,
\begin{align*}
(0.5 + \gamma) \cdot \lp(a) & > \stlb(a) = \sum_i \stlb_i(a) \ge 0.5 \cdot \sum_{i \in U_0} \lp_i(a) + (0.5 + \gamma^{\nicefrac{3}{4}}) \cdot \sum_{i \notin U_0} \lp_i(a) \\
& = (0.5 + \gamma^{\nicefrac{3}{4}}) \cdot \lp(a) - \gamma^{\nicefrac{3}{4}} \cdot \sum_{i \in U_0}\lp_i(a) = (0.5 + \gamma^{\nicefrac{3}{4}}) \cdot \lp(a) - \gamma^{\nicefrac{3}{4}} \cdot \lp(\tilde{a}),
\end{align*}
where the first inequality follows from the assumption of the lemma, and the second inequality follows from Lemma~\ref{lem:algbase} and the definition of $U_0$. Rearranging the inequality finishes the proof.

For the second part of the statement, it trivially holds for all $i \notin U_0$, since all $a_{it}=0$.

For $i \in U_0$, by applying Lemma~\ref{lem:free_prob} to $\tilde a$ with parameters $\mu =\gamma^{\nicefrac{3}{4}}$ and $\beta_1=\alpha$, we have that
\begin{flalign*}
& 1) \sum_t \tilde a_{it}^L = \sum_t \tilde a_{it} \cdot \one [\wit > \alpha \cdot \lp_i(\tilde a)] \le \sqrt{\frac{4\gamma^{\nicefrac{3}{4}}}{\alpha-0.5-\gamma^{\nicefrac{3}{4}}}} \le \sqrt{\frac{4\gamma^{\nicefrac{3}{4}}}{1-0.5-0.001}} < \gamma^{\nicefrac{1}{4}}~, & \\
& 2) \lp_i(\tilde a^L) = \sum_t \tilde a_{it}^L \cdot \wit \le \frac{0.5+\gamma^{\nicefrac{3}{4}}}{1-3\gamma^{\nicefrac{3}{8}}} \cdot \lp_i(\tilde a) \le 0.5 + \gamma^{1/4}~.
\end{flalign*}
Finally, by applying Lemma~\ref{lem:free_prob} to $\tilde a$ with parameters $\mu =\gamma^{\nicefrac{3}{4}}$ and $\beta_2=0.5+\gamma^{\nicefrac{1}{4}}+\gamma^{\nicefrac{3}{4}}$, we have that
\[
\sum_t \tilde a_{it} \cdot \one [w_{it} > \beta \cdot \lp_i(\tilde a)] \le \sqrt{\frac{4\gamma^{\nicefrac{3}{4}}}{\beta_2-0.5-\gamma^{\nicefrac{3}{4}}}} = \gamma^{\nicefrac{1}{4}}~.
\]
Consequently, we have
\begin{align*}
\lp_i(\tilde a) & = \sum_{t: w_{it}/\lp_i(\tilde a) > \alpha} \tilde a_{it} \cdot \wit + \sum_{t: w_{it}/\lp_i(\tilde a) \in (\beta_2,\alpha]} \tilde a_{it} \cdot \wit + \sum_{t: w_{it}/\lp_i(\tilde a) \le \beta_2} \tilde a_{it} \cdot \wit \\
& \le \lp_i(\tilde a^L) + \alpha \gamma^{\nicefrac{1}{4}} \cdot \lp_i(\tilde a) + \beta_2 \cdot \lp_i(\tilde a) 
\end{align*}
Rearranging the inequality gives that
\[
\lp_i(\tilde a^L) \ge (1-\alpha \cdot \gamma^{\nicefrac{1}{4}} - \beta_2) \cdot \lp_i(\tilde a) \ge (0.5-(\alpha+2) \gamma^{\nicefrac{1}{4}}) \cdot \lp(\tilde a)~. \qedhere 
\]
\end{proof}

%% file: slackness.tex
\section{Main Algorithm and Proof Overview}

We now present our main algorithm for proving Theorem~\ref{thm:main}. Starting from this section, for the ease of notation, we normalize
\[
\lpexante ~=~ 1~.
\]

Our algorithm consists of three parameters $\eps, \epsphilo, \epsslack$, that are constants to be optimized.
Our target is to have a $(0.5+\eps)$ order-competitive ratio. The meaning of the other two parameters shall be clear soon.

\paragraph{Step 1:}
As the first step, we solve the ex-ante relaxation and check if the following inequality holds:
\begin{equation}
\label{eqn:checkbase}
\stlb(x^*) \ge 0.5+\eps~.
\end{equation}
Recall that $\lpexante$ is renormalized as $1$. Therefore, if \eqref{eqn:checkbase} is satisfied, the baseline algorithm $\algbase(x^*)$ is at least $(0.5+\eps)$-competitive and we simply run this algorithm. Observe that the quantity $\stlb(x^*)$ can be efficiently calculated, and hence, this procedure is efficient.

\paragraph{Step 2:}
From now on, we assume that
\begin{equation}
\label{eqn:basenot}
\stlb(x^*) < 0.5+\eps~.
\end{equation}

Applying Lemma~\ref{lma:frdt} to $x^*$ and $\alpha=2$, we calculate $\tilde x$ and $\tilde x^L$ with the following properties. 
\begin{itemize}
    \item $\lp(\tilde x) \ge (1-\eps^{1/4}) $~;
    \item $q_i(\tildxl) = \sum_t \tilde x_{it}^L \le \probfr$ for every $i \in U$, where we define $\probfr ~:=~ \eps^{1/4}$.
    \item $(0.5 - 4 \eps^{\nicefrac{1}{4}}) \cdot \lp_i(\tilde{x}) \leq \lp_i(\tilde x^L) \leq (0.5 + \eps^{\nicefrac{1}{4}} ) \cdot \lp_i(\tilde x)$, which implies $(0.5 - 4.5\eps^{1/4}) \leq \lp(\tildxl) \leq 0.5 + \eps^{1/4}$.
\end{itemize}
For the rest of the paper, we use $L$ to denote the edge set 
\begin{equation}
L \eqdef \{ (i,t) \in E: \wit \ge 2 \cdot \lp_i(\tilde x)\}~.
\end{equation}
This notation is consistent with $\tilde x^L$ as it can be interpreted as $\tilde x_{it}^L = \tilde x_{it} \cdot \one[(i,t) \in L]$.

Each edge in $L$ is referred to as the ``free'' edges of our instance, while each $(i, t) \in E \setminus L$ is considered a ``deterministic'' edge. These two concepts correspond to the ``free'' and ``deterministic'' vertices discussed in Section~\ref{sec:warmup}. As we extend the vertex-weighted setting from Section~\ref{sec:warmup} to the edge-weighted setting, the definitions of ``free'' and ``deterministic'' vertices are similarly generalized to ``free'' and ``deterministic'' edges.


\paragraph{Step 3:}
In step 1, we checked whether $\stlb(x^*)$ is large enough so that the baseline algorithm $\algbase(x^*)$ can achieve a competitive ratio strictly greater than $0.5$. 
A natural extension is to examine the following statement:
\begin{equation}
\label{eqn:check}	
\exists x, \quad \stlb(x) > 0.5+\eps~.
\end{equation}
If so, we could simply apply the baseline algorithm with the corresponding $x$; else, it should provide more information of the instance. 
Informally, the instance should satisfy an appropriate notion of uniqueness (assumptions 5 of Section~\ref{sec:warmup}) when \eqref{eqn:check} fails.
Unfortunately, it is computationally inefficient to verify \eqref{eqn:check}, since $\stlb(x)$ is generally non-concave in $x$. To resolve this issue, we introduce the following \emph{slackness LP}. 
\begin{align*}
\max_{y}: \quad & \sum_{(i,t) \in E} \wit \cdot \tilde x^L_{it} \cdot \left(1 - \frac{\yit}{p_t} \right) + \sum_{(i,t) \in L} \wit \cdot \yit \cdot \left(1 - \frac{\tilde x^L_{it}}{p_t} \right) \tag{$\lpslack$} \\
\text{subject to:} \quad & \sum_{(i,t) \in E} \wit \cdot \yit \ge 1-\epsphilo \notag \\
& \sum_{i} \yit \le p_t & \forall t \in R \notag \\
& \sum_{t} \yit \le 1 & \forall i \in L \notag
\end{align*}

Intuitively speaking, the slackness LP verifies if there is an $\epsphilo$-approximately optimal solution $y$ with a large distance from $\tilde x$ on larges edges $L$, where $\epsphilo$ is the second parameter of our algorithm. Note that
\[
\lpslack ~\geq~ \sum_{(i,t) \in L} \wit \cdot (\tildxl_{it} - \yit)^+ + \sum_{(i,t) \in L} \wit \cdot (\yit - \tildxl_{it})^+ ~=~ \sum_{(i,t) \in L} \wit \cdot |\tildxl_{it} - \yit|,
\]
where notation $r^+$ represents $\max\{r, 0\}$. Therefore, when $\lpslack$ is small, it implies either $\lp(y^*)$ is bounded by $1 - \epsphilo$, or the decisions made by $y^*$ and $\tildx$ on free edges are almost the same, where we recall $y^*$ is the corresponding solution of $\lpexante$ for the optimal online algorithm.

Motivated by the above observation, depending on the value of $\lpslack$ and using $\epsslack$ as the cutting point, we branches with two different algorithms. The exact use of the slackness LP shall be highlighted in the next two sections.

\paragraph{Modeling the assumptions in Section~\ref{sec:warmup}.} We have addressed all the five assumptions from Section~\ref{sec:warmup} through the slackness LP. Before continuing to present the branched algorithms, we provide a brief summary of how these assumptions are modeled in our main algorithm.
\begin{itemize}
    \item Assumption 1 (free/deterministic vertices), Assumption 2 (vertex-weighted), and Assumption 4 (balance between free and deterministic parts): These assumptions are addressed by our free-deterministic decomposition. When applying \Cref{lma:frdt} to the solution $x^*$, the free/deterministic vertices assumption is generalized to free/deterministic edges (edge sets $L$ and $E \setminus L$), and the vertex-weighted assumption is extended to edge weights. The balance between the free and deterministic parts is ensured by \Cref{lma:frdt}.
    \item Assumption 3 ($\philo = \lpexante$): We make use of this assumption by only considering solutions $y$ with $\lp(y) \geq 1 - \epsphilo$. If $\lp(y^*) < 1 - \epsphilo$, we will show later in \Cref{cor:small_slack} that running $\algbase(x^*)$ in this case is sufficient.
    \item Assumption 5 (uniqueness): We extend the uniqueness assumption by requiring that the solution $y^*$ is unique for free edges. This holds when $\lpslack$ is bounded by $\epsslack$, as the weighted difference between $\tildx$ and $y^*$ on free edges is limited by $\epsslack$ (see the discussion above). When $\lpslack > \epsslack$, we will demonstrate below in \Cref{thm:large_slack} that a separated algorithm can identify a solution $z$ such that $\stlb(z)$ is at least $0.5 + \eps$.
\end{itemize}


\medskip

Now, we continue to discuss our branched algorithm, which is based on whether $\lpslack \leq \epsslack$.

\paragraph{Step 3a, Large Slackness.}
 When $\lpslack$ is large (at least $\epsslack$), we show the existence of a good baseline algorithm, i.e., we prove \eqref{eqn:check}. To be specific, we give the following \Cref{thm:large_slack}.
\begin{restatable}{Theorem}{thmlargeslack}
\label{thm:large_slack}
If $\lpslack \ge \epsslack$,   there exists $z \in P$ with $\stlb(z) \ge 0.5+ \eps$ when $10^{-4} \geq \epsphilo \geq 3\eps$ and $\epsslack \ge 150(\eps^{1/4}+\epsphilo^{1/4})$ are both satisfied.

Furthermore, the solution $z$ can be computed efficiently.
\end{restatable}

Let $y^o$ be the optimal solution of $\lpslack$. Intuitively, when $\lpslack$ is large, $\tildx$ and $y^o$ are both high-value solutions, but make different decisions on edge set $L$. We make use of this difference to properly merge solutions $\tildx$ and $y^o$ into $z$, and show that $\stlb(z)$ is sufficiently large. The detailed proof of \Cref{thm:large_slack} is provided in Section~\ref{sec:large_slackness}.

\paragraph{Step 3b, Small Slackness.}
The challenging case is when $\lpslack$ is small (at most $\epsslack$). Indeed, this is the only place where our algorithm adapts to the arrival order. And we are now in a similar situation as in the warm-up section. 
We generalize the algorithm of Section~\ref{sec:warmup} so that it works well when $\philo$ is sufficiently close to $\lpexante = 1$. The proof of the following theorem is provided in Section~\ref{sec:small_slackness}.

\begin{restatable}{Theorem}{thmsmallslack}
\label{thm:small_slack}
If $\lpslack < \epsslack$, there exists an order-unaware algorithm $\alg$ that achieves an expected reward of $0.625 - 1.5 \epsslack^{1/3} - \epsphilo - 6\eps^{1/4} - \epsslack$, as long as the optimal order-aware algorithm is $\epsphilo$-approximate to the ex-ante optimum, i.e., when $\philo \ge 1-\epsphilo$.
\end{restatable}
We defer the detailed proof of \Cref{thm:small_slack} to Section~\ref{sec:small_slackness}, and first show how to use \Cref{thm:small_slack} to give a $0.5 + \eps$ order-competitive algorithm. Recall that \Cref{thm:large_slack} assumes $\epsslack \geq \Omega(\eps^{1/4} + \epsphilo^{1/4})$. With this extra assumption, we give the following corollary:

\begin{Corollary}
\label{cor:small_slack}
If $\lpslack < \epsslack$, while we have $10^{-4} \geq \epsslack \geq 150(\eps^{1/4} + \epsphilo^{1/4})$, there exists an order-unaware algorithm with order-competitive ratio $0.5 + (\nicefrac{1}{16}-\epsslack^{\nicefrac{1}{3}}) \cdot \epsphilo$. 
\end{Corollary}

\begin{proof}
Let 
\[
c =  0.625 - 1.5 \epsslack^{1/3} - \epsphilo - 6\eps^{1/4} - \epsslack - 0.5 ~\geq~ 0.125 - 2\epsslack^{1/3},
\]
where the second inequality holds when  $10^{-4} \geq \epsslack \geq 150(\eps^{1/4} + \epsphilo^{1/4})$ is satisfied. We run the algorithm $\alg$ from Theorem~\ref{thm:small_slack} with probability $\delta_{\alg}$, and $\algbase(x^*)$ with probability $1-\delta_{\alg}$. This is a randomized order-unaware algorithm. Note that
\begin{itemize}
\item If $\philo < 1-\epsphilo$, the expected reward of our algorithm is at least
\[
(1-\delta_{\alg}) \cdot 0.5  ~\ge~ 0.5 \cdot \frac{1-\delta_{\alg}}{1-\epsphilo} \cdot \philo~.
\]
\item If $\philo \ge 1-\epsphilo$, the expected reward of our algorithm is at least
\[
\delta_{\alg} \cdot (0.5 + c)  + (1-\delta_{\alg}) \cdot 0.5  = 0.5 + c \cdot \delta_{\alg} ~\geq~ (0.5 + c \cdot \delta_{\alg}) \cdot \philo
\]
\end{itemize}

By setting $\delta_{\alg} = \frac{\epsphilo}{1+2c(1-\epsphilo)} $, the randomized algorithm guarantees an order competitive ratio of 
\[
\Gamma = 0.5 + \frac{c \cdot \epsphilo}{1+2c \cdot (1-\epsphilo)} ~\geq~ 0.5 + \left (\frac{1}{16} - \epsslack^{1/3} \right) \cdot \epsphilo,
\]
where we use $c \geq 0.125 - 2\epsslack^{1/3}$ and $1 + 2c\cdot (1 - \epsphilo) \leq 1 + 2\cdot 0.125 < 2$ in the last inequality.
\end{proof}


\subsection{Proof of Theorem~\ref{thm:main}}
\label{sec:ratio}
Finally, we prove \Cref{thm:main} by combining \Cref{thm:large_slack} and \Cref{cor:small_slack}.

We set $\eps = 10^{-27}$, $\epsphilo = 256 \cdot 10^{-27}$, and $\epsslack = 10^{-6} \geq 150(\eps^{1/4} + \epsphilo^{1/4}) = 7.5 \cdot 10^{-7}$. Since the condition for \Cref{thm:large_slack} is satisfied, we have a feasible solution $z$ of $\lpexante$, such that $\algbase(z) \geq \stlb(z) \geq 0.5 + \eps$, which implies a $(0.5 + \eps)$ order-competitive algorithm when $\lpslack \geq \epsslack$. On the other hand, when $\lpslack < \epsslack$, \Cref{cor:small_slack} guarantees the order competitive ratio is at least $0.5 + (\frac{1}{16} - \epsslack^{1/3}) \cdot \epsphilo \geq 0.5 + \eps$. Therefore, we have a $(0.5 + \eps)$ order-competitive algorithm in both cases, which proves \Cref{thm:main}.

\paragraph{Why is $\eps$ so small.} We remark that the constants in our analysis are not fully optimized. We believe the value of $\eps$ can be improved to approximately $10^{-20}$ by providing a more careful analysis. However, from \Cref{thm:small_slack}, we require $O(\epsslack^{1/3}) \leq 0.125$ to achieve a competitive ratio greater than $0.5$, and \Cref{thm:large_slack} implies that $\epsslack \geq \Omega(\eps^{1/4})$. We remark that the exponent $1/4$ from our free-deterministic decomposition and the exponent $1/3$ from our algorithm for the small slackness case are both tight. Consequently, the maximum possible $\eps$ following the current analysis is bounded by $(0.125)^{12}$. Our ratio suffers an additional significant constant loss from \Cref{thm:large_slack}. 

%% file: small_slack.tex
\section{Small slackness}
\label{sec:small_slackness}

In this section, we study the case when the slackness of instance is small and prove \Cref{thm:small_slack}, which is restated below for convenience:

\thmsmallslack*

\subsection{Our Algorithm}

Our algorithm consists of two components: 1) an online algorithm that maintains a \emph{dynamic} fractional solution $\{\xit\}$ of the instance; 2) an online rounding algorithm applying to our fractional solution $\{\xit\}$.  The formal description of our small slackness algorithm is provided below. It is a generalization of the algorithm described in Section~\ref{sec:warmup}.
\vspace{0.5cm}

\begin{tcolorbox}[title=Algorithm when $\lpcore$ is small] 
\label{alg:smallslack}
Initialization:
\begin{enumerate}[(1)]
    \item Initiate $\xit = \tildxl_{it}$ for every $(i,t) \in E$. 
    \item Initiate $t_i = \infty$ for every $i \in U$. If $t \le t_i$, we say vertex $\ui$ is in the first stage, else we say $\ui$ is in the second stage. Initially, all offline vertices are in the first stage.
    \item For notation simplicity, we introduce a dummy offline vertex $0$ with $w_{0t} = 0$ for every $t \in V$, and let 
    $x_{0t} = p_t - \sum_{i} \xit~$.
    During the algorithm, condition $q_t(x) = p_t$ is always satisfied.
\end{enumerate}
\vspace{0.5cm}

Upon the arrival of $\vt$:
\begin{enumerate}[Step 1:]
    \item Let $\vt$ propose to $\ui$ with probability $\frac{\xit}{p_t}$ if $\vt$ is realized.
    \item If $\ui$ receives a proposal from $\vt$:
        \begin{itemize}
            \item if $t_i = \infty$, greedily select the edge $(i,t)$ if $\ui$ is unmatched at time $t$;
            \item Otherwise, select the edge $(i,t)$ if $\ui$ is unmatched at time $t$ with probability:
            \[
            \frac{1}{2 (1 - \nicefrac{1}{2} \cdot \sum_{t_i < s < t} \xit[s])}~.
            \]
        \end{itemize}
    \item \textbf{Update the stages of offline vertices:} \\
    For every $i$ in stage $1$ (i.e., $t_i = \infty$), set $t_i = t$ and claim $i$ is moved to stage $2$ iff 
    \begin{equation}
    \label{eqn:switch_stage}
    \sum_{s \le t} \wit \cdot \tildxl_{is} \ge (1-\epsalg) \cdot \sum_{s} \wit \cdot \tildxl_{is}~,	
    \end{equation}
    where $\epsalg$ is a parameter to be optimized.
    \end{enumerate}
\end{tcolorbox}
\clearpage

\begin{tcolorbox}[title=Algorithm when $\lpcore$ is small (cont'd)] 
\begin{enumerate}[Step 4:]
    \item \textbf{Update the weights for stage $2$ vertices:} \\
    For each $t_i = t$:
    \begin{enumerate}
        \item Sort edges $\{(j, s)\}$ that satisfy $(i, s) \notin L$, $x_{js} > 0$, and $s > t$ in descending order based on the value of $\widehat{w}_{is} - \widehat{w}_{js}$, where $\widehat{w}_{js} = 2 \cdot w_{js}$ for $(j,s) \in L$ and $\widehat{w}_{js} = w_{js}$ for $(j,s) \notin L$.
        \item For each edge $(j, s)$ in the above sorted list, let 
        \[
        \textstyle \Delta_{i,j,s} = \min\{1 - \probfr -  \sum_{s': (i, s') \notin L} x_{is'}, x_{js}\}~.
        \]. Update $x_{is} \gets \Delta_{i,j,s}$ and $x_{js} \gets x_{js} - \Delta_{i,j,s}$.
        \item Stop the above loop when  $\sum_{s': (i,s') \notin L} x_{is'}$ reaches $1 - \probfr$.
    \end{enumerate}
\end{enumerate}
\end{tcolorbox}

\subsection{Analysis of the Algorithm}

To begin, we introduce some extra notations that shall be useful for our analysis.
\begin{itemize}
    \item For every $(i,s) \in E$ and $t \in V$, we use $x^{(t)}_{is}$ to denote the value of $x_{is}$ when $t$ arrives. 
    Note the following properties for $x^{(t)}_{is}$:
    \begin{itemize}
        \item For $(i, s) \in L$, $x^{(t)}_{is}$ is non-increasing with respect to $t$, i.e., $x^{(t)}_{i,s} \leq \tildxl_{is}$ for every $t$. 
        \item For $(i, s) \notin L$, $x^{(t)}_{is} \leq x^{(t_i)}_{is}$. 
        \item The above two properties imply $\sum_{s \in V} x^{(t)}_{is} \leq \sum_{s: (i, s) \notin L} x^{(t_i)}_{is} + \probfr \leq 1$. Therefore, in Step 2 the matching probability is feasible since $ \frac{1}{2 (1 - \nicefrac{1}{2} \cdot \sum_{t_i < s < t} \xit[s])}$ is bounded by $1$.
    \end{itemize}
    \item We define $E_1 = \{(i,t) : t \le t_i \}$ and $E_2 = \{(i,t) : t > t_i \}$ as a partition of the set of edges, i.e., $(i, t) \in E_1$ implies that vertex $i$ is still in stage $1$; $(i, t) \in E_2$ implies that $i$ moves to stage $2$.
\end{itemize}

Next, we provide the step-by-step analysis of the algorithm.

\paragraph{Online rounding. (Step 1 and 2)}
Recall that in the warm up section (the deterministic and free vertices setting), there is no component of rounding as we abstract out the random realization of online vertices.
Our rounding algorithm is rather straightforward. Upon the arrival of an online vertex $t$, it proposes to at most one of its neighbors, with probability $\nicefrac{\xit}{p_t}$ each.
Then, each offline vertex accept proposals greedily in stage $1$, and runs an $\nicefrac{1}{2}$-selectable online contention resolution scheme (OCRS) in stage $2$. The following lemma characterizes the loss of our rounding algorithm.

\begin{restatable}{Lemma}{lmarounding}
\label{lma:rounding}
\[
\alg \ge (1-\probfr) \cdot \left( \sum_{(i,t) \in E_1} \wit \cdot x^{(t)}_{it} + \frac{1}{2} \cdot \sum_{(i,t) \in E_2} \wit \cdot x^{(t)}_{it} \right)
\]
\end{restatable}

Intuitively, our algorithm collects almost entire  value of $w_{it} \cdot x^{(t)}_{it}$ when $i$ is in stage $1$, as $i$ is only matched to free edges. Therefore, the probability that $i$ got matched in stage $1$ should be bounded by $\probfr$. When $i$ moves to stage $2$, the OCRS framework guarantees that half of the value $w_{it} \cdot x^{(t)}_{it}$ is collected. We deferred the detailed proof of \Cref{lma:rounding} to \Cref{sec:lmarounding-proof}.

\paragraph{Fractional solution. (Step 3 and 4)}

We are left to analyze our fractional solution $x$. We shall prove that our fractional algorithm achieves a total reward of at least $0.625 - o(1)$.

Our analysis proceeds similarly as in Section~\ref{sec:warmup}. 
For each offline vertex, we keep collecting free edges until $1-\epsalg$ fraction of the free edges have arrived (refer to equation~\eqref{eqn:switch_stage}). According to the slackness check of our algorithm, since the slackness of the optimal online algorithm $y^*$ is bounded by $\epsslack$, intuitively, $y^*$ should have the same behavior as solution $\tildx$, i.e., the optimal algorithm also only collect free edges in stage $1$, implying that each offline vertex is still free with high probability when switching to stage $2$, and in stage $2$ the optimal algorithm should mainly collect reward from deterministic edges.  As a result, the optimal algorithm collects a sufficiently high reward stage $2$ even if we exclude those large-value edges in $L$. Formally, we give the following \Cref{lem:opt_second}:
\begin{Lemma}
\label{lem:opt_second}
    $\sum_{(i,t) \in E_2 \setminus L} \wit \cdot \yit^* \ge  0.5 - \epsphilo - \eps^{1/4} - \epsslack -4 \sqrt{\frac{\epsslack}{\epsalg}}$.
\end{Lemma}

\Cref{lem:opt_second} is where we uses the fact that the slackness of $y^*$ is bounded by $\epsslack$. We defer the proof of \Cref{lem:opt_second} to \Cref{sec:lemoptsecond-proof}.

Finally, we interpret the Step $4$ of our algorithm as an ``online weighted matching with free disposal'' algorithm for solving the below online problem. Define
\[
\widehat w_{it} = \begin{cases} 2\cdot w_{it} & (i,t) \in L  \\ w_{it} & (i,t) \in E_2 \setminus L \\ 0 & (i,t) \in E_1 \setminus L
\end{cases}~.
\]
Consider the following way of understand our algorithm:
\begin{itemize}
\item Initially, only the set of large edges $L$ is given at the beginning, each with weight $\widehat w_{it}$. Each vertex $\ui \in U$ has capacity $1$ and each vertex $\vt \in V$ has capacity $p_t$. The matching is initiated as $x = \tildxl$.
\item At time $t_i$ (in ascending order of $t_i$), the edge set $\{(i,t)\}_{t\in V} \cap (E_2 \setminus L$ is observed , each with weight $\widehat w_{it}$. The algorithm (fractionally) matches the newly observed immediately. Previously selected edges (including those larges edges provided at the beginning) can be freely disposed of, while they cannot be retrieved once we proceeds to the next time step.
\item The goal of Step $4$ is to maximize the total increment over the original matching, under an extra constraint that the total probability on the newly observed edges (i.e., $\sum_{t: (i,t) \in E_2 \setminus L} x^{(t)}_{it}$) is bounded by $1 - \probfr$. 
\end{itemize}

Note that if we proceed Step $4$ in an offline manner, the optimal increment of the above problem can be characterized as the following linear program.
\begin{align*}
\max: \quad & \sum_{(i,t) \in E} \widehat w_{it} \cdot (z_{it} - \tildxl_{it}) & ~ \\
\text{subject to}: \quad & \sum_{t:(i,t) \in E_2 \setminus L} z_{it}  \le 1 - \probfr & \forall \ui \in U \\
& \sum_{i:(i,t) \in E_2 \setminus L} z_{it} + \sum_{i:(i,t) \in L} z_{it} \le p_t & \forall \vt \in V \\
& z_{it} \le \tildxl_{it} & \forall (i,t) \in L
\end{align*}

As the Step $4$ in our algorithm is essentially a greedy algorithm that maximizes the marginal increment, the total increment gained by our algorithm should be at least $\frac{1}{2} \cdot \sum_{(i,t) \in E} \widehat w_{it} \cdot (\tildz_{it} - \tildxl_{it})$ for any feasible solution $\{z_{it}\}$. Note that
\[
z_{it} = \begin{cases} \min(\tildxl_{it}, \yit^*) & (i,t) \in L  \\ (1-\probfr) \cdot \yit^* & (i,t) \in E_2 \setminus L \\ 0 & (i,t) \in E_1 \setminus L
\end{cases}~
\]
is a feasible solution of the above LP. Then, the following lemma suggests that the total increment of our algorithm is at least one half of the increment given by the above solution:

\begin{restatable}{Lemma}{lmasubmod}
    \label{lma:submod}
   We have
    \[
    \sum_{(i,t) \in E} \widehat w_{it} \cdot x^{(t)}_{it} - \sum_{(i,t) \in L} \widehat w_{it} \cdot \tildxl_{it}  ~\geq~ \frac{1}{2} \cdot \left((1 - \probfr) \cdot \sum_{(i,t) \in E_2 \setminus L}  \widehat w_{it} \cdot y^*_{it} - \sum_{(i,t) \in L} \widehat w_{it} \cdot (\tildxl_{it} - y^*_{it})^+\right),
    \]
    where we use notation $r^+$ to denote $\max\{r, 0\}$.
\end{restatable}

We defer the proof of \Cref{lma:submod} to \Cref{sec:lmasubmod-proof}.

With the help of \Cref{lma:rounding}, \Cref{lem:opt_second} and \Cref{lma:submod}, we are ready to complete the proof of Theorem~\ref{thm:small_slack}.

\begin{proof}[Proof of \Cref{thm:small_slack}]
Recall that \Cref{lma:rounding} gives
    \begin{align}
\label{eqn:alg_uniq_small}
\alg & \ge  (1-\probfr) \cdot \left( \sum_{(i,t) \in E_1} \wit \cdot x^{(t)}_{it} + \frac{1}{2} \cdot \sum_{(i,t) \in E_2} \wit \cdot x^{(t)}_{it} \right) \notag \\
& \ge (1-\probfr) \cdot \left( \sum_{(i,t) \in E_1} \wit \cdot \tildxl_{it} + \frac{1}{2} \cdot \sum_{(i,t) \in E_2 \setminus L} \wit \cdot x^{(t)}_{it} + \sum_{(i,t) \in L} \wit \cdot (x^{(t)}_{it} -\tildxl_{it}) \right) \notag \\
&= (1-\probfr) \cdot \left( \sum_{(i,t) \in E_1} \wit \cdot \tildxl_{it} + \frac{1}{2} \cdot \left(\sum_{(i,t) \in E} \widehat w_{it} \cdot x^{(t)}_{it} - \sum_{(i,t) \in L} \widehat w_{it} \cdot \tildxl_{it} \right)\right) \notag \\
& \ge (1-\probfr) \cdot \left( \sum_{(i,t) \in E_1} \wit \cdot \tildxl_{it} + \frac{1}{4} \cdot \left((1 - \probfr) \cdot \sum_{(i,t) \in E_2 \setminus L}  \widehat w_{it} \cdot y^*_{it} - \sum_{(i,t) \in L} \widehat w_{it} \cdot (\tildxl_{it} - y^*_{it})^+\right) \right)~, 
\end{align}
where the first inequality holds by Lemma~\ref{lma:rounding}, the second inequality uses the fact that in stage $1$ we only collect rewards from $(i, t) \in E_1 \cap L$ and the fact that $x^{(t)}_{it} \leq \tildxl_{it}$ for $(i,t) \in E_2 \cap L$, the third equality uses the fact that $x^{(t)}_{it} = 0$ for $(i,t) \in E_1 \setminus L$ and the definition of $\widehat w_{it}$, and the last inequality holds by \Cref{lma:submod}.

Now, we bound the terms in \eqref{eqn:alg_uniq_small} separately. We first have
\begin{align*}
    \sum_{(i,t) \in E_1} \wit \cdot \tildxl_{it} ~\geq~ (1 - \epsalg) \cdot \lp(\tildxl),
\end{align*}
because the algorithm collects at least $1 - \epsalg$ fraction of the reward from free edges in stage $1$. Next, \Cref{lem:opt_second} guarantees
\begin{align*}
    \sum_{(i,t) \in E_2 \setminus L}  \widehat w_{it} \cdot y^*_{it} ~=~ \sum_{(i,t) \in E_2 \setminus L}  w_{it} \cdot y^*_{it} ~\geq~ 0.5 - \epsphilo - \eps^{1/4} - \epsslack -4 \sqrt{\frac{\epsslack}{\epsalg}}.
\end{align*}
For the last term in \eqref{eqn:alg_uniq_small}, we have
\begin{align*}
    \sum_{(i,t) \in L} \widehat w_{it} \cdot (\tildxl_{it} - y^*_{it})^+ ~=~ 2 \sum_{(i,t) \in L} w_{it} \cdot (\tildxl_{it} - y^*_{it})^+ ~\leq~ 2\sum_{(i,t) \in L} w_{it} \cdot \tildxl_{it} \cdot \left(1 - \frac{y^*_{it}}{p_t}\right) ~\leq~ 2\epsslack.
\end{align*}

Applying the above three inequalities to \eqref{eqn:alg_uniq_small}, we have
\begin{align*}
    \alg ~&\geq~ (1-\probfr) \cdot \left( (1-\epsalg) \cdot \lp(\tildxl) + \frac{1}{4} \cdot \left( 0.5 - \epsphilo - \eps^{1/4} - \epsslack -4 \sqrt{\frac{\epsslack}{\epsalg}} - \probfr - 2\epsslack \right) \right) \\
    ~&\geq~ 0.625 - \probfr - \frac{\epsalg}{2} - 4.5\epsilon^{1/4} - \frac{\epsphilo + \eps^{1/4} + 3\epsslack}{4} - \sqrt{\frac{\epsslack}{\epsalg}}  \\
    ~&\geq~ 0.625 - 1.5 \epsslack^{1/3} - \epsphilo - 6\eps^{1/4} - \epsslack,
\end{align*}
where the last inequality is satisfied when taking $\epsalg = \epsslack^{1/3}$.
\end{proof}

\subsection{Proof of \Cref{lem:opt_second}}
\label{sec:lemoptsecond-proof}

Let $U_0 \subseteq U$ be the set of vertices that the optimal algorithm matches with probability less than $\deltnew$, i.e.,
\[
U_0 \eqdef \left\{ \ui \in U: \sum_{t < t_i} \yit^* \le \deltnew \right\} \quad \text{and} \quad \overline{U_0} \eqdef U \setminus U_0~,
\]
where $\deltnew$ is defined to be $\sqrt{\epsslack/\epsalg}$.
Since the realization of online vertex $t$ is independent to the event that an vertex $i$ is matched before $t$, we have that
\[
\yit^* \le p_t \cdot \left( 1 - \sum_{s<t} \yit[s]^* \right), \quad \forall \ui \in U, \vt \in V~.
\]
We remark that this is the crucial constraint of previous works~\cite{SW-icalp21,BDL-EC22,NSW-arxiv23,BDPSW-arxiv24} for approximating the online optimum with known arrival orders.
Applying this inequality to $t \ge t_i$, we have that
\[
\frac{\yit^*}{p_t} \le 1 - \sum_{s<t} \yit[s]^* \le 1 - \sum_{s< t_i} \yit[s]^* \le 1-\deltnew , \quad \forall i \in \overline{U_0}, t \ge t_i~.
\]
Furthermore, since $\{\yit^*\}$ is a feasible solution of $\lpcore$, we have that
\begin{align*}
\epsslack \ge \lpcore & \ge \sum_{i,t} \wit \cdot \tildxl_{it} \cdot \left(1 - \frac{\yit^*}{p_t} \right) + \sum_{L} \wit \cdot \yit^* \cdot \left(1 - \frac{\tildxl_{it}}{p_t} \right) \\
& \ge \sum_{i \in \overline{U_0}} \sum_{t \ge t_i} \wit \cdot \tildxl_{it} \cdot \left(1 - \frac{\yit^*}{p_t} \right) + \sum_{L} \wit \cdot \yit^* \cdot \left(1 - \frac{\tildxl_{it}}{p_t} \right) \\
& \ge \deltnew \cdot \epsalg \cdot \sum_{i \in \overline{U_0}} \lp_i(\tildxl) + \sum_{L} \wit \cdot (\yit^* - \tildxl_{it})^+,
\end{align*}
where the last inequality uses the fact that online vertices $t \geq t_i$ contributes at least $\epsalg$ fraction of $\lp_i(\tildxl)$. Therefore, the above inequality gives
\begin{align*}
    \sum_{i \in \overline{U_0}} \lp_i(\tildxl) ~\leq~ \frac{\epsslack}{\deltnew \cdot \epsalg} \qquad \text{and} \qquad \sum_{L} \wit \cdot (\yit^* - \tildxl_{it})^+ ~\leq~ \epsslack.
\end{align*}

Now, we are ready to lower-bound $\sum_{(i,t)\in E_2 \setminus L} \wit \cdot \yit^*$. We have
\begin{align}
\label{eqn:philo_second}
\sum_{(i,t)\in E_2 \setminus L} \wit \cdot \yit^* &= \philo - \sum_{(i,t)\in E_1 \cup L} \wit \cdot \yit^* \notag \\
&\ge \philo - \sum_{(i,t) \in L} \wit \cdot \yit^* - \sum_{ (i,t) \in E_1 \setminus L \land i \in U_0} \wit \cdot \yit^* - \sum_{(i,t) \in E_1 \setminus L \land i \in \overline{U_0}} \wit \cdot \yit^* 
\end{align}

It remains to bound the three summations in \eqref{eqn:philo_second} separately. For the first summation, we have
\begin{align}
    \sum_{L} \wit \cdot \yit^* ~\leq~ \sum_{L} \wit \cdot \tildxl_{it} + \sum_{L} \wit \cdot (\yit^* - \tildxl_{it})^+ ~\leq~ \lp(\tildxl) + \epsslack. \label{eq:firstsum}
\end{align}
For the second summation, we have
\begin{align}
    \sum_{ (i,t) \in E_1 \setminus L \land i \in U_0} \wit \cdot \yit^* ~\leq~ \sum_{i \in U_0} \sum_{t: (i, t) \in E_1 \setminus L} 2\lp_i(\tildx) \cdot y^*_{it} ~\leq~ 2\deltnew \cdot \lp(\tildx), \label{eq:secondsum}
\end{align}
where the first inequality uses the fact that $w_{it} \leq 2\lp_i(x)$ for $(i,t) \notin L$, and the second inequality uses the fact that for every $i \in U_0$, we have
\[
\sum_{t:(i,t) \in E_1 \setminus L} y^*_{it} ~\leq~ \sum_{t < t_i} y^*_{it} ~\leq~ \deltnew.
\]
Here, we use a critical observation that $(i, t_i) \in L$.

For the third summation in \eqref{eqn:philo_second}, we have
\begin{align}
    \sum_{(i,t) \in E_1 \setminus L \land i \in \overline{U_0}} \wit \cdot \yit^*  ~\leq~ \sum_{i \in \overline{U_0}} \sum_{t: (i, t) \in E_1 \setminus L} 2\lp_i(\tildx) \cdot y^*_{it} ~\leq~ 2\sum_{i \in \overline{U_0}} \lp_i(\tildx) ~\leq~ \frac{2\epsslack}{\deltnew \cdot \epsalg}, \label{eq:thirdsum}
\end{align}
where the inequality uses the fact that $w_{it} \leq 2\lp_i(x)$ for $(i,t) \notin L$. 

Finally, applying the above \eqref{eq:firstsum}, \eqref{eq:secondsum} and \eqref{eq:thirdsum} to \eqref{eqn:philo_second} together with the assumption that $\philo \geq 1 - \epsphilo$ gives
\[
\sum_{(i,t)\in E_2 \setminus L} \wit \cdot \yit^* ~\geq~ 1 - \epsphilo - \lp(\tildxl) - \epsslack - 2\deltnew - \frac{2\epsslack}{\deltnew \cdot \epsalg} ~\geq~ 0.5 - \epsphilo - \eps^{1/4} - \epsslack - 4\sqrt{\frac{\epsslack}{\epsalg}},
\]
where the last inequality uses the fact that $\lp(\tildxl)  \leq 0.5 + \eps^{1/4}$ and that $\deltnew := \sqrt{\epsslack/\epsalg}$.

%% file: large_slack.tex
\section{Large Slackness}
\label{sec:large_slackness}

In this section, we give the proof of \Cref{thm:large_slack}, which is restated below for convenience:

\thmlargeslack*

\subsection{Proof of \Cref{thm:large_slack}}

Let $y^o$ be the optimal solution of $\lpslack$. Since we assume $\lpslack \geq \epsslack$, we have
\[
\sum_{(i,t) \in E} \wit \cdot \tilde x^L_{it} \cdot \left(1 - \frac{y^o_{it}}{p_t} \right) + \sum_{(i,t) \in L} \wit \cdot y^o_{it} \cdot \left(1 - \frac{\tilde x^L_{it}}{p_t} \right) ~\geq~ \epsslack.
\]

The first step of our algorithm is to check whether $\stlb(y^o) \geq 0.5 + \eps$. If so, then returning $y^o$ as the desired solution is sufficient. Otherwise, we apply \Cref{lma:frdt} to solution $y^o$. Recall that $\lpslack$ guarantees that $\lp(y^o) \geq 1 - \epsphilo$. Therefore, we have $\stlb(y^o) \leq \lp(y^o) \cdot \frac{0.5 + \eps}{1 - \epsphilo} \leq 0.5 + \epsphilo$, where the last inequality holds when $\epsphilo \geq 3\eps$ and $\epsphilo \leq 10^{-4}$. Then, applying \Cref{lma:frdt} with $\gamma = \epsphilo$ and $\alpha = 1$ gives
\begin{itemize}
    \item $\lp(\tilde y) \ge (1-\epsphilo^{1/4}) \cdot \lp(y^o) \geq  1-\epsphilo^{1/4}-\epsphilo $
    \item $q_i(\tilde y^L) = \sum_t \tilde y_{it}^L \le \epsphilo^{1/4}$ for every $i \in U$
    \item $\tildyl_{it} = \one[w_{it} > \lp_i(\tilde y)] \cdot \tildy_{it}$
    \item $\lp(\tilde y^L)  \geq (0.5 - 3 \epsphilo^{\nicefrac{1}{4}}) \cdot \lp(\tilde{y}) \geq  0.5 - 3.5 \epsphilo^{\nicefrac{1}{4}}$, which further implies $\lp(y^L) \geq 0.5 - 3.5 \epsphilo^{1/4} - \epsphilo$.
\end{itemize}

After decomposing $y^o$ into $\tildy$ and $\tildyl$, we first show that $\tildy$ is also with a large slackness. Note that
\begin{align}
    &\sum_{(i,t) \in L(\tildx)} \wit \cdot \tildxl_{it} \cdot \left(1 - \frac{\tildy_{it}}{p_t} \right) + \sum_{(i,t) \in L(\tildx)} \wit \cdot \tildy_{it} \cdot \left(1 - \frac{\tildxl_{it}}{p_t} \right) \notag \\
    ~\geq~&\sum_{(i,t) \in L(\tildx)} \wit \cdot \tildxl_{it} \cdot \left(1 - \frac{y^o_{it}}{p_t} \right) + \sum_{(i,t) \in L(\tildx)} \wit \cdot y^o_{it} \cdot \left(1 - \frac{\tildxl_{it}}{p_t} \right) - \sum_{(i,t) \in L(\tildx)} \wit \cdot (y^o_{it} - \tildy_{it}) \notag \\
    ~\geq~& \epsslack - (\lp(y^o) - \lp(\tildy)) ~\geq~ \epsslack - \epsphilo^{1/4}, \label{eq:yslack}
\end{align}

where the second line uses the fact that $\tildy_{it} \leq y^o_{it}$, and the third line uses the assumption that the slackness of $y^o$ is at least $\epsslack$. Therefore, the slackness of solution $\tildy$ is also sufficiently large. 

We further decompose \eqref{eq:yslack} into two parts:  
\begin{align}
    &\sum_{(i,t) \in L(\tildx)} \wit \cdot \tildxl_{it} \cdot \left(1 - \frac{\tildy_{it}}{p_t} \right) + \sum_{(i,t) \in L(\tildx)} \wit \cdot \tildy_{it} \cdot \left(1 - \frac{\tildxl_{it}}{p_t} \right)  \notag \\
    ~\leq~& \sum_{(i,t)} \wit \cdot \tildxl_{it} \cdot \left(1 - \frac{\tildyl_{it}}{p_t} \right) + \sum_{(i,t)} \wit \cdot \tildyl_{it} \cdot \left(1 - \frac{\tildxl_{it}}{p_t} \right) + \sum_{\substack{(i,t): \\ 2\lp_i(\tildx) < w_{it} \leq \lp_i(\tildy)}} \wit \cdot \tildy_{it} \cdot \left(1 - \frac{\tildxl_{it}}{p_t} \right) \notag \\
    ~\leq~& \sum_{(i,t)} \wit \cdot \tildxl_{it} \cdot \left(1 - \frac{\tildyl_{it}}{p_t} \right) + \sum_{(i,t)} \wit \cdot \tildyl_{it} \cdot \left(1 - \frac{\tildxl_{it}}{p_t} \right) + \sum_{i: 2\lp_i(\tildx) <  \lp_i(\tildy)} \lp_i(\tildy) \notag \\
    ~\leq~& \sum_{(i,t)} \wit \cdot \tildxl_{it} \cdot \left(1 - \frac{\tildyl_{it}}{p_t} \right) + \sum_{(i,t)} \wit \cdot \tildyl_{it} \cdot \left(1 - \frac{\tildxl_{it}}{p_t} \right) \notag\\
    &+ 2\cdot \sum_{i\in U}  \left(\lp_i(\tildy) - \lp_i(\tildx)\right) \cdot \one[2\lp_i(\tildx) < \lp_i(\tildy)] \notag,
\end{align}
 where the first inequality uses the fact that $\tildyl_{it} \leq \tildy_{it}$ and the fact that $\tildyl_{it} = \tildy_{it} \cdot \one[w_{it} > \lp_i(\tildy)]$, the second inequality follows from the observation that the last term sums over edges such that $w_{it} \leq \lp_i(\tildy)$, while $\sum_t \tildy_{it} \leq 1$ for every $i$. Combining the above inequality with \eqref{eq:yslack} implies that either we have
\begin{align}
    \sum_{i,t} w_{it} \cdot \tildxl_{it} \cdot \left(1 - \frac{\tildyl_{it}}{p_t}\right) + \sum_{i,t} w_{it} \cdot \tildyl_{it} \cdot \left(1 - \frac{\tildxl_{it}}{p_t}\right) ~\geq~ 0.6 \cdot (\epsslack - \epsphilo^{1/4}), \label{eq:slack1}
\end{align}
 
or we have
\begin{align}
     \sum_{i\in U}  \left(\lp_i(\tildy) - \lp_i(\tildx)\right) \cdot \one[2\lp_i(\tildx) < \lp_i(\tildy)] ~\geq~ 0.2 \cdot (\epsslack - \epsphilo^{1/4}). \label{eq:slack2}
\end{align}

Then, the following two lemmas suggest that we can find a good solution in either case:

\begin{Lemma}
\label{lma:largerelax}
    Suppose we are given $\tildx, \tildy, \tildxl, \tildyl$ that are generated by \Cref{lma:frdt}. Assume the following inequality holds:
    \[
    \sum_{i,t} w_{it} \cdot \tildxl_{it} \cdot \left(1 - \frac{\tildyl_{it}}{p_t}\right) + \sum_{i,t} w_{it} \cdot \tildyl_{it} \cdot \left(1 - \frac{\tildxl_{it}}{p_t}\right) \geq \yifana.
    \]
    When $\epsphilo, \eps \leq 10^{-4}$ and $\epsphilo \geq 3\eps$ are satisfied, there exists a solution $a$, such that $\stlb(a) \geq \min\{0.5 + \eps, 0.5 + \eps_a\}$, where $\eps_a$ is a constant that satisfies 
    \[
        \eps_a \geq \frac{\yifana}{32}  - 2.5\eps^{1/4} - 2.5 \epsphilo^{1/4}.
    \]
    Furthermore, solution
     $a$ can be computed efficiently.
\end{Lemma}

\begin{Lemma}
\label{lma:largenorm}
Suppose we are given $\tildx, \tildy, \tildxl, \tildyl$ that are generated by \Cref{lma:frdt}. Assume the following inequality holds:
   \[
   \sum_{i \in U} \left(\lp_i(\tildy) - \lp_i(\tildx) \right) \cdot \one[\lp_i(\tildy) \geq 2\lp_i(\tildx)] \geq \yifanb
   \]
    When $\epsphilo, \eps \leq 10^{-4}$ and $\epsphilo \geq 3\eps$ are satisfied, there exists a solution $b$ such that $\stlb(b) \geq \min\{0.5 + \eps, 0.5 + \eps_b\}$, where $\eps_b$ is a constant that satisfies 
    \[
    \eps_b \geq \frac{1}{4}\yifanb - 6\epsphilo^{1/4} -6\eps^{1/4}.
    \]
    Furthermore, solution $b$ can be computed efficiently.
\end{Lemma}

We use \Cref{lma:largerelax} and \Cref{lma:largenorm} to prove \Cref{thm:large_slack}. It can be verified that when $\epsslack \geq 150(\eps^{1/4} + \epsphilo^{1/4})$ and $\eps, \epsphilo \leq 10^{-4}$, both 
\[
\frac{0.6(\epsslack - \epsphilo^{1/4})}{32} - 2.5(\epsphilo^{1/4} + \epsslack^{1/4}) \geq \eps \qquad \text{and} \qquad \frac{0.2(\epsslack - \epsphilo^{1/4})}{4} - 6(\epsphilo^{1/4} + \epsslack^{1/4}) \geq \eps
\]
hold. Therefore, we receive at least one solution $z$ that satisfies $\stlb(z) \geq 0.5 + \eps$ from either \Cref{lma:largerelax} or \Cref{lma:largenorm},  which proves \Cref{thm:large_slack}.

\subsection{Proof of \Cref{lma:largerelax}}
\label{sec:largerelax}

Let $\tilde{a} = \frac{\tildx+\tildy}{2}$ and $\tilde{a}^L = \frac{\tildxl+\tildyl}{2}$. Define solution $\bar{a}$ to be the following: for every $i \in U, t \in V$, we set 
\[
\bar{a}_{it} = \min\left\{2, \frac{p_t}{q_t(\tilde{a}^L)}\right\} \cdot \tilde{a}^L_{it}. 
\]

Consider the following two cases:

\paragraph{Case 1: $\lp(\bar{a}) \ge \frac{0.5+\eps}{1-\probfr - \epsphilo^{1/4}}$.} Note that $q_i(\bar{a}) \leq 2 \cdot \frac{q_i(\tildxl) + q_i(\tildyl)}{2} \leq \probfr + \epsphilo^{1/4}$. Then, if $\lp(\bar{a}) \ge \frac{0.5+\eps}{1-\probfr - \epsphilo^{1/4}}$, there must be
\[
\stlb(\bar{a}) ~\geq~ \sum_{i \in U} (1 - \probfr - \epsphilo^{1/4}) \cdot \lp_i(\bar{a}) ~\geq~ (1 - \probfr - \epsphilo^{1/4}) \cdot \lp(\bar{a}) ~\geq~ 0.5 + \eps.
\]
Therefore, $\bar{a}$ would be the desired solution for \Cref{lma:largerelax} when  $\lp(\bar{a}) \ge \frac{0.5+\eps}{1-\probfr - \epsphilo^{1/4}}$.

\paragraph{Case 2: $\lp(\bar{a}) < \frac{0.5+\eps}{1-\probfr - \epsphilo^{1/4}}$.} In this case, define $U_1, U_2$ to be two subsets of offline vertices. For each $i \in U$, let $i$ be included in $U_1,U_2$ with probability $0.5$ respectively. Now consider the following solution $a$ (note the difference between $\tilde{a}^L_{it}$ and $\tilde{a}_{it}$):
\begin{align*}
& a_{it}~:=~ \tilde{a}^L_{it} + \sum_{j \in U_2} \tilde{a}^L_{jt} \cdot \frac{\tilde{a}^L_{it}}{p_t}, \quad \forall i \in U_1 \\
& a_{it} ~:=~ \tilde{a}_{it} - \sum_{j \in U_1} \tilde{a}^L_{jt} \cdot \frac{\tilde{a}^L_{it}}{p_t}, \quad \forall i \in U_2 \\
\end{align*}
To check the feasibility of solution $a$ (whether $a \in P$), note that for $i \in U_1$, we have
\begin{align*}
    a_{it} ~\leq~ \tilde{a}^L_{it}  + \tilde{a}^L_{it} \cdot \frac{\sum_{j \in U} \tilde{a}^L_{jt}}{p_t} ~\leq~ 2\tilde{a}^L_{it},
\end{align*}
so $q_i(a) \leq \probfr + \epsphilo^{1/4} \leq 1$, and for $i \in U_2$, we have $q_i(a) \leq q_i(\tilde{a}) \leq 1$. For $q_t(a)$, we have
\begin{align*}
    q_t(a) ~&=~ \sum_{i \in U_1} a_{it} + \sum_{i \in U_2} a_{it} ~=~ \sum_{i \in U} \tilde{a}^L_{it} + \sum_{i \in U_1} \sum_{j \in U_2} \tilde{a}^L_{jt} \cdot \frac{\tilde{a}^L_{it}}{p_t} - \sum_{i \in U_2} \sum_{j \in U_1} \tilde{a}^L_{jt} \cdot \frac{\tilde{a}^L_{it}}{p_t}  ~=~ q_t(\tilde{a}^L) ~\leq~ p_t.
\end{align*}
Therefore, we have $a \in P$. 

Now consider the value of $\stlb(a)$, we have
\begin{align}
\E[\stlb(a)] & \ge \E\left[\sum_{i \in U_1} \stlb_i(a)+ \sum_{i \in U_2}\stlb_i(a)\right] \notag \\
& \ge \E\left[ \sum_{i \in U_1} (1-\probfr - \epsphilo^{1/4}) \cdot \lp_i(a) + \frac{1}{2} \sum_{i \in U_2} \lp_i(a)\right] \notag\\
& = \E\left[ \sum_{i \in U_1} (1-\probfr - \epsphilo^{1/4}) \cdot \lp_i(\tilde{a}^L) + \frac{1}{2} \sum_{i \in U_2} \lp_i(\tilde{a}) + (0.5 - \probfr - \epsphilo^{1/4}) \cdot \sum_{i \in U_1} \sum_{j \in U_2} \sum_{t \in V} \frac{\tilde{a}^L_{it} \cdot \tilde{a}^L_{jt}}{p_t} \cdot \wit \right]\notag \\
& = \frac{1}{2} \cdot (1-\probfr - \epsphilo^{1/4}) \cdot \lp(\tilde{a}^L) + \frac{1}{4} \cdot \lp(\tilde{a}) + (0.5 - \probfr - \epsphilo^{1/4}) \cdot \frac{1}{4} \cdot \sum_{i \in U} \sum_{j \neq i} \sum_{t \in V} \frac{\tilde{a}^L_{it} \cdot \tilde{a}^L_{jt}}{p_t} \cdot \wit, \label{eq:largerelax1}
\end{align}
where the second line is from the observation that $q_i(a) \leq \probfr + \epsphilo^{1/4}$ for $i \in U_1$.

Note that for every $t$, we have
\begin{align}
    \sum_{i \in U} \sum_{j\ne i} \frac{\tilde{a}^L_{it} \cdot \tilde{a}^L_{jt}}{p_t} \cdot \wit ~&=~ \sum_{i \in U} \tilde{a}^L_{it} \cdot w_{it} \cdot \left(\frac{q_t(\tilde{a}^L)}{p_t}  - \frac{\tilde{a}^L_{it}}{p_t} + 1 - 1 \right) \notag \\
    ~&=~ \sum_{i \in U} \tilde{a}^L_{it} \cdot w_{it} \cdot  \left(1 - \frac{\tilde{a}^L_{it}}{p_t} \right) + \sum_{i \in U} \tilde{a}^L_{it} \cdot w_{it} \cdot  \left(\frac{q_t(\tilde{a}^L)}{p_t} - 1 \right) \label{eq:largerelax2}
\end{align}
For the first term of \eqref{eq:largerelax2}, note that
\begin{align*}
    \tilde{a}^L_{it} \cdot \left(1 - \frac{\tilde{a}^L_{it}}{p_t} \right) ~&=~ \frac{\tildxl_{it} + \tildyl_{it}}{2} \cdot \left(1 - \frac{\tildxl_{it} + \tildyl_{it}}{2p_t} \right) \\
    ~&=~ \frac{\tildxl_{it} + \tildyl_{it}}{4} + \frac{p_t \cdot (\tildxl_{it} + \tildyl_{it}) - (\tildxl_{it} + \tildyl_{it})^2}{4p_t} \\
    ~&\geq~\frac{\tildxl_{it} + \tildyl_{it}}{4} + \frac{(\tildxl_{it})^2 + (\tildyl_{it})^2 - (\tildxl_{it} + \tildyl_{it})^2}{4p_t} \\
    ~&=~ \frac{1}{4} \cdot \tildxl_{it} \cdot \left (1 - \frac{\tildyl_{it}}{p_t}\right) + \frac{1}{4} \cdot \tildyl_{it} \cdot \left (1 - \frac{\tildxl_{it}}{p_t}\right).
\end{align*}
Therefore,
\begin{align}
    \sum_{i \in U} \tilde{a}^L_{it} \cdot w_{it} \cdot  \left(1 - \frac{\tilde{a}^L_{it}}{p_t} \right) ~\geq~ \frac{1}{4} \cdot \sum_{i \in U} \tildxl_{it} \cdot w_{it} \cdot \left(1 - \frac{\tildyl_{it}}{p_t}\right) + \frac{1}{4} \cdot \sum_{i \in U} \tildyl_{it} \cdot w_{it} \cdot  \left(1 - \frac{\tildxl_{it}}{p_t}\right). \label{eq:largerelax3}
\end{align}

For the second term of \eqref{eq:largerelax2}, note that 
\begin{align*}
    \tilde{a}^L_{it} \cdot w_{it} \cdot  \left(\frac{q_t(\tilde{a}^L)}{p_t} - 1 \right) ~\geq~  \tilde{a}^L_{it} \cdot w_{it} \cdot  \left(1 - \min\left\{2, \frac{p_t}{q_t(\tilde{a}^L)}\right\}\right),
\end{align*}
where the inequality follows from the fact that $q_t(\tilde{a}^L)/p_t + p_t/q_t(\tilde{a}^L) \geq 2$. Therefore,
\begin{align}
    \sum_{i \in U} \tilde{a}^L_{it} \cdot w_{it} \cdot  \left(\frac{q_t(\tilde{a}^L)}{p_t} - 1 \right) ~&\geq~ \sum_{i \in U} \tilde{a}^L_{it} \cdot w_{it}  - \sum_{i \in U} \tilde{a}^L_{it} \cdot w_{it} \cdot  \min\left\{2, \frac{p_t}{q_t(\tilde{a}^L)}\right\}  \notag \\
    ~&=~ \sum_{i \in U} \tilde{a}^L_{it} \cdot w_{it}  - \sum_{i \in U} \bar{a}_{it} \cdot w_{it} \label{eq:largerelax4}
\end{align}

Applying \eqref{eq:largerelax3} and \eqref{eq:largerelax4}  to \eqref{eq:largerelax2} and summing the inequality over all $t \in V$, we have
\begin{align*}
    \sum_{i \in U} \sum_{j \neq i} \sum_{t \in V} \frac{\tilde{a}^L_{it} \cdot \tilde{a}^L_{jt}}{p_t} \cdot \wit ~\geq~ \frac{1}{4} \cdot \yifana + \lp(\tilde{a}^L) - \lp(\bar{a}).
\end{align*}
Finally, applying the above inequality  to \eqref{eq:largerelax1}, we have
\begin{align*}
    \E[\stlb(a)] ~\geq~ \frac{5}{8} \lp(\tilde{a}^L) + \frac{1}{4} \lp(\tilde{a}) + \frac{1}{32} \yifana - \frac{1}{8} \lp(\bar{a}) - \frac{3}{4}\probfr - \frac{3}{4}\epsphilo^{1/4} .
\end{align*}
Applying the fact that $a = \frac{x+y}{2}$ and $\tilde{a}^L = \frac{\tildxl + \tildyl}{2}$, we have 
\[
\E[\stlb(a)] ~\geq~ \frac{1}{8} \big(\lp(\tildx) + \lp(\tildy)\big) + \frac{5}{16} \big(\lp(\tildxl) + \lp(\tildyl)\big) + \frac{1}{32} \yifana - \frac{1}{8} \lp(\bar{a}) - \frac{3}{4} \probfr - \frac{3}{4} \epsphilo^{1/4}.
\]
Finally, applying the bounds for $\lp(\tildx), \lp(\tildy), \lp(\tildxl), \lp(\tildyl)$ and the assumption $\lp(\bar{a}) < \frac{0.5 + \eps}{1 - \probfr - \epsphilo^{1/4}}$ to the above inequality shows that $\E[\stlb(a)] $ is at least $0.5 + \eps_a$, such that
\begin{align*}
    \eps_a 
    &\geq \frac{\yifana}{32} -\frac18 \cdot (\eps^{1/4} + \epsphilo^{1/4}+ \epsphilo) - \frac5{16}  \cdot (3.5\epsphilo^{1/4} + \epsphilo+  4.5\eps^{1/4}
    )  -\frac18 \cdot (\eps + \probfr + \epsphilo^{1/4} ) - \frac{3}{4} \probfr - \frac{3}{4} \epsphilo^{1/4}\\
    &\geq \frac{\yifana}{32}  - 2.5\eps^{1/4} - 2.5 \epsphilo^{1/4},
\end{align*}
where we use in the first inequality that $\lp(\bar a)   \leq \frac{0.5 + \eps}{1 - \probfr - \epsphilo^{1/4}} \leq 0.5 + \eps + \probfr + \epsphilo^{1/4}$, which is true when $\epsphilo, \eps \leq 10^{-4}$, and the second inequality uses the assumption that $\probfr = \eps^{1/4}$ and $\epsphilo, \eps \leq 10^{-4}$, which allows us to further upper bound the extra $O(\epsphilo)$ and $O(\eps)$ terms by $O(\epsphilo^{1/4})$ and $O(\eps^{1/4})$ respectively.

\subsection{Proof of \Cref{lma:largenorm}}
\label{sec:largenorm}

We first construct solution $\tilde{b}$ via the following process:
\begin{itemize}
    \item Initiate $\tilde{b}= \tildyl$.
    \item For each $t$ such that $q_t(\tilde{b}) > q_t(\tildxl)$, arbitrarily reduce the value of $\tilde{b}_{it}$ until $q_t(\tilde{b}) = q_t(\tildxl)$ is satisfied.
\end{itemize}

The generated $z$ satisfies $q_t(\tilde{b}) \leq q_t(\tildxl)$ for every $t \in V$. Let solution $\bar{b} = \tildxl + \tildyl - \tilde{b}$. Then, solution $\bar{b}$ satsifies $\bar{b} \in $, because we have  $q_i(\bar{b}) \leq q_i(\tildxl) + q_i(\tildyl) \leq \probfr + \epsphilo^{1/4}$ for every $i \in U$, and $q_t(\bar{b}) \leq q_t(\tildyl) \leq p_t$ for every $t \in V$. Consider the following two cases:

\paragraph{Case 1: $\lp(\bar{b}) \ge \frac{0.5+\eps}{1-\probfr - \epsphilo^{1/4}}$.} Note that $q_i(\bar{b}) \leq 2 \cdot \frac{q_i(\tildxl) + q_i(\tildyl)}{2} \leq \probfr + \epsphilo^{1/4}$. Then, if $\lp(\bar{b}) \ge \frac{0.5+\eps}{1-\probfr - \epsphilo^{1/4}}$, there must be
\[
\stlb(\bar{b}) ~\geq~ \sum_{i \in U} (1 - \probfr - \epsphilo^{1/4}) \cdot \lp_i(\bar{b}) ~\geq~ (1 - \probfr - \epsphilo^{1/4}) \cdot \lp(\bar{b}) ~\geq~ 0.5 + \eps.
\]
Therefore, $\bar{b}$ would be the desired solution for \Cref{lma:largenorm} when  $\lp(\bar{b}) \ge \frac{0.5+\eps}{1-\probfr - \epsphilo^{1/4}}$.

\paragraph{Case 2: $\lp(\bar{b}) < \frac{0.5+\eps}{1-\probfr - \epsphilo^{1/4}}$.} In this case, consider the solution 
\[
b = (1 - \epsphilo^{1/4}) \cdot (\tildx - \tildxl) + \tilde{b}.
\]
Solution $b$ is feasible because for $i \in U$, we have 
\[
q_i(b) ~\leq~ (1 - \epsphilo^{1/4}) \cdot q_i(\tildx) + q_i(\tilde{b}) ~\leq~ 1 - \epsphilo^{1/4} + q_i(\tildyl) ~\leq~ 1,
\]
On the other hand, for $t \in V$, we have
\[
q_t(b) ~=~ (1 - \epsphilo^{1/4}) \cdot (q_t(\tildx) - q_t(\tildxl)) + q_t(\tilde{b}) ~\leq~ q_t(\tildx) - q_t(\tildxl) + q_t(\tildxl) ~\leq~ p_t.
\]
Therefore, we have $b \in P$.

Now consider the value of $\stlb(b)$. Let $U_3 = \{i: 2\lp_i(\tildx) \leq \lp_i(\tildy)\}$ and $U_4 = \{i: 2\lp_i(\tildx) > \lp_i(\tildy)\}$. For $i \in U_3$, a feasible baseline algorithm would be setting a threshold $\lp_i(\tildy)$. Since $\tildx_{it} - \tildxl_{it} > 0$ only when $w_{it} < 2\lp_i(\tildx) \leq \lp_i(\tildy)$, setting the threshold to be $\lp_i(\tildy)$ could only take values when $\tilde{b}_{it} > 0$. Note that $q_i(\tilde{b}) \leq q_i(\tildyl) \leq \epsphilo^{1/4}$, so setting the threshold to be $\lp_i(\tildy)$ gets at least $(1 - \epsphilo^{1/4}) \cdot \lp_i(\tilde{b})$. Therefore,
\begin{align}
    \stlb(b) ~\geq&~ \sum_{i \in U_3} (1 - \epsphilo^{1/4}) \lp_i(\tilde{b}) + \sum_{i \in U_4} \frac{1}{2} \cdot \lp_i(b) \notag \\
    ~\geq&~ \sum_{i \in U_3} \frac{1}{2}\cdot \lp_i(b)  + \sum_{i \in U_3} \frac{1}{2} \cdot \big( \lp_i(\tilde{b}) + \lp_i(\tildxl) - \lp_i(\tildx)\big) + \sum_{i \in U_4} \frac{1}{2} \cdot \lp_i(b) - \epsphilo^{1/4} \notag \\
    ~\geq&~ \frac{1}{2} \lp(b) - \epsphilo^{1/4} + \frac{1}{2} \sum_{i \in U_3} \big(\lp_i(\tildyl) - (0.5 + 4\eps^{1/4}) \lp_i(\tildx)\big) + \frac{1}{2} \sum_{i \in U_3} \big(\lp_i(\tilde{b}) - \lp_i(\tildyl)\big),\label{eq:largenorm1}
\end{align}
where the second inequality uses the fact that $\tilde{b} = \frac{1}{2} b + \frac{1}{2} \tilde{b} - \frac{1 - \epsphilo^{1/4}}{2} \cdot (\tildx - \tildxl)$, and the third inequality uses the fact that $\lp_i(\tildxl) \geq (0.5-4\eps^{1/4}) \cdot \lp_i(\tildx)$.

We further bound the last two terms of \eqref{eq:largenorm1}. For the third term, we have
\begin{align}
    \sum_{i \in U_3} \big(\lp_i(\tildyl) - (0.5 + 4\eps^{1/4}) \lp_i(\tildx)\big) ~&\geq~ \sum_{i \in U_3} \big((0.5 - 3\epsphilo^{1/4})\lp_i(\tildy) - (0.5 + 4\eps^{1/4}) \lp_i(\tildx)\big)\notag  \\
    ~&\geq~ \frac{1}{2} \sum_{i \in U_3} \big(\lp_i(\tildy) - \lp_i(\tildx) \big) - 3\epsphilo^{1/4} - 4\eps^{1/4} \notag \\
    ~&\geq~ \frac{1}{2}\yifanb - 3\epsphilo^{1/4} - 4\eps^{1/4}, \label{eq:largenorm2}
\end{align}
where the last inequality follows from the assumption in \Cref{lma:largenorm}. For the last term of \eqref{eq:largenorm1}, we have
\begin{align}
    \sum_{i \in U_3} \big(\lp_i(\tilde{b}) - \lp_i(\tildyl)\big) ~&\geq~ \sum_{i \in U} \big(\lp_i(\tilde{b}) - \lp_i(\tildyl)\big) \notag \\
    ~&=~ \lp(\tilde{b}) - \lp(\tildyl) ~=~ \lp(\tildxl) - \lp(\bar{b}) \label{eq:largenorm3}
\end{align}
Now apply \eqref{eq:largenorm2} and \eqref{eq:largenorm3} back to \eqref{eq:largenorm1}, we have
\begin{align*}
    \stlb(b) ~\geq&~~ \frac{1}{2} \lp(b) - \epsphilo^{1/4} + \frac{1}{4} \yifanb - \frac32 \epsphilo^{1/4} - 2\probfr + \frac{1}{2} \lp(\tildxl) - \frac{1}{2} \lp(\bar{b}) \\
    ~=&~~ \frac{1}{2} (1 - \epsphilo^{1/4}) \cdot \big(\lp(\tildx) - \lp(\tildxl)\big) + \frac{1}{2}\big( \lp(\tildxl) + \lp(\tildyl) - \lp(\bar{b}) \big) \\
    &+ \frac{1}{2} \lp(\tildxl) - \frac{1}{2} \lp(\bar{b}) + \frac{1}{4} \yifanb - \frac52 \epsphilo^{1/4} - 2\eps^{1/4} \\
    ~\geq&~~ \frac{1}{2} \big(\lp(\tildx) + \lp(\tildxl) + \lp(\tildyl) - 2\lp(\bar{b})\big)  + \frac{1}{4} \yifanb - 3 \epsphilo^{1/4} - 2\eps^{1/4} ,
\end{align*}
where the second line uses the definition of $b$ and the fact that $\tilde{b}= \tildxl + \tildyl - \bar{b}$.
Finally, applying the bounds for $\lp(\tildx), \lp(\tildxl), \lp(\tildyl)$ and the assumption $\lp(\bar{b}) < \frac{0.5 + \eps}{1 - \probfr - \epsphilo^{1/4}}$ to the above inequality gives $\stlb(b) \geq 0.5 + \eps_b$, where
\begin{align*}
    \eps_b 
    ~&\geq~\frac{1}{4}\yifanb - \frac12 (3.5 \epsphilo^{1/4} + \epsphilo+ 5.5\eps^{1/4}) - (\eps +\probfr + \epsphilo^{1/4}) - 3 \epsphilo^{1/4} - 2\eps^{1/4} \\
    ~&=~\frac{1}{4}\yifanb - \frac{23}{4}\epsphilo^{1/4} - \frac{23}{4}\probfr - \eps - \epsphilo\\
    ~&\geq~\frac{1}{4}\yifanb - 6\epsphilo^{1/4} -6\eps^{1/4},
\end{align*}
where we use in the first inequality that $\lp(\bar a) \leq \frac{0.5 + \eps}{1 - \probfr - \epsphilo^{1/4}} \leq  0.5 + \eps + \probfr + \epsphilo^{1/4}$, which is true when $\epsphilo, \eps \leq 10^{-4}$, and the second inequality uses the assumption that $\probfr = \eps^{1/4}$ and $\epsphilo, \eps \leq 10^{-4}$, which allows us to further upper bound the extra $O(\epsphilo)$ and $O(\eps)$ terms by $O(\epsphilo^{1/4})$ and $O(\eps^{1/4})$ respectively.

%% file: smallslack_appendix.tex
\section{Missing Proofs in Section \ref{sec:small_slackness}}

\subsection{Proof of \Cref{lma:rounding}}
\label{sec:lmarounding-proof}

\lmarounding*

\begin{proof}
We consider the contribution of each offline vertex $i$ to our algorithm separately.
Let $r_i = \Pr{i \text{ is matched in stage } 1}$. We first show that $r_i \le \probfr$:
\begin{align*}
r_i = \Pr{i \text{ is matched in stage }1} ~=&~ \Pr{i \text{ is proposed by some }t \le t_i} \\
=&~ 1 - \prod_{t \le t_i} (1-x^{(t)}_{it}) \le \sum_{t \le t_i} x^{(t)}_{it} \le \sum_{t \le t_i} \tildxl_{it} \le \probfr~,
\end{align*}
where we use the fact that $x^{(t)}_{it}$ is at most the initial value of $x_{it}$ when $(i,t) \in L$.
Then, the expected matching of $\ui$ from the first stage is
\begin{align*}
&\sum_{t\le t_i} \wit \cdot \Pr{\ui \text{ is proposed by } \vt } \cdot \Pr{i \text{ is unmatched before }t} \\
\ge~& \sum_{t\le t_i} \wit \cdot \Pr{\ui \text{ is proposed by } \vt } \cdot \Pr{i \text{ is unmatched at }t_i} \\
\ge~& \sum_{t\le t_i} \wit \cdot x^{(t)}_{it} \cdot (1 - \probfr) = (1-\probfr) \cdot \sum_{t: (i,t) \in E_1} \wit \cdot x^{(t)}_{it}~.
\end{align*}
Next, we consider the expected reward from the second stage. We prove by induction that every edge $(i,t)$ from the second stage (i.e., $t>t_i$) is matched with probability $\frac{1-r_i}{2}$. 
Suppose the statement holds for every $t_i < s <t$. Consider the edge $(i,t)$. According to the design of our algorithm, $(i,t)$ is selected with probability:
\begin{multline*}
\frac{1}{2 (1 - \frac{1}{2} \cdot \sum_{t_i < s < t} x^{(s)}_{is})} \cdot \Pr{\ui \text{ is proposed by } \vt } \cdot \Pr{i \text{ is unmatched before }t} \\
= \frac{1}{2 (1 - \frac{1}{2} \cdot \sum_{t_i < s < t} x^{(s)}_{is})} \cdot x^{(t)}_{it} \cdot \left(1 - r_i - \frac{1-r_i}{2} \cdot \sum_{t_i < s < t} x^{(s)}_{is} \right) = \frac{1-r_i}{2} \cdot x^{(t)}_{it}~,
\end{multline*}
where the first equation holds by induction hypothesis.
Therefore, the contribution from the second stage equals
\[
\sum_{t>t_i} \wit \cdot \frac{1-r_i}{2} \cdot \xit \ge (1-\probfr) \cdot \frac{1}{2} \cdot \sum_{t:(i,t) \in E_2} \wit \cdot \xit~.
\]
Summing the above equations over all $i \in U$ concludes the proof of the lemma.
\end{proof}

\subsection{Proof of \Cref{lma:submod}}
\label{sec:lmasubmod-proof}

\lmasubmod*

\begin{proof}
    Recall that we modeled the Step $4$ of our algorithm as an ``online matching with free disposal'' process, while the main idea of \Cref{lma:submod} is that the increment of the matching value of greedy algorithm achieves at least one-half times the optimal offline increment. However, we find it involved to prove \Cref{lma:submod} via the online matching with free disposal process. To further simplify the proof, we view the Step $4$ of our algorithm as a more general ``online submodular welfare maximization'' process, and show that greedy algorithm is a $\frac{1}{2}$-approximation.

To be specific,  consider to view $i \in U$ as ``online items'', and $t \in V$ are ``offline agents''.
Without loss of generality, we assume that offline vertices $i \in U$ are relabeled according to the time of stage transition, i.e., for $i, j \in U$, we assume $t_i \leq t_{j}$ for $i < j$. Furthermore, when $t_i = t_j$ while $i < j$, we assume $i$ is proceeded in Step $4$ for transitioning to stage $2$ earlier than $j$. 
Then, the order of $i$ represents the arrival order of online items in our submodular welfare maximization process.

For the submodular welfare function, we define the following function $f_t(\{r_{it}\}_{(i,t) \notin L}; \{\tildxl_{it}\}_{(i,t) \in L}\})$. The value of $f_t(\{r_{it}\}_{(i,t) \notin L}; \{\tildxl_{it}\}_{(i,t) \in L}\})$ is defined as the objective of the following LP:
\begin{align*}
\max: \quad & \sum_{(i,t) \in E} \widehat w_{it} \cdot z_{it} - \sum_{(i,t) \in L} \widehat w_{it} \cdot \tildxl_{it} & \lpsubmod \\
\text{subject to}: \quad &  \sum_{i \in U} z_{it}  ~\le~ p_t &  \\
& z_{it} ~\le~ \tildxl_{it} & \forall i:(i,t) \in L \\
& z_{it} ~\le~ r_{it} & \forall i:(i,t) \notin L,
\end{align*}
where we recall $\widehat w_{it}$ is defined as 
\[
\widehat w_{it} = \begin{cases} 2\cdot w_{it} & (i,t) \in L  \\ w_{it} & (i,t) \in E_2 \setminus L \\ 0 & (i,t) \in E_1 \setminus L
\end{cases}~.
\]

For simplicity, we will abbreviate the notation  $f_t(\{r_{it}\}_{(i,t) \notin L}; \{\tildxl_{it}\}_{(i,t) \in L}\})$ as $f_t(\{r_{it}\})$ in the following proof.   Note that function $f_t(\{r_{it}\})$ is monotone and non-negative, where the non-negativity is guaranteed by the observation that setting $z_{it} = \tildxl_{it}$ is a feasible solution of $\lpsubmod$. Furthermore, the following claim shows that $f_t$ is a submodular function:

\begin{Claim}
    \label{clm:submod}
    Function $f_t(\{r_{it}\})$ is submodular, that is, given $\{r_{it}\}$ and $\{r'_{it}\}$ such that for every $(i, t) \notin L$ we have $0 \leq r_{it} \leq r'_{it}$, for $j: (j,t) \notin L$ and $\Delta_r \geq 0$, we have
    \[
    f_t\left(\{r_{it}\}_{-j} \cup \{r_{jt} + \Delta_r\}\} \right) - f_t\left(\{r_{it}\} \right) ~\geq~ f_t\left(\{r'_{it}\}_{-j} \cup \{r'_{jt} + \Delta_r\}\} \right) - f_t\left(\{r'_{it}\} \right),
    \]
    where the notations $\{r_{it}\}_{-j}$ and $\{r_{it}\}_{-j}$ represent $\{r_{it}\}_{i: (i,t) \notin L \land i \neq j}$ and $\{r_{it}\}_{i: (i,t) \notin L \land i \neq j}$ respectively.
\end{Claim}

Intuitively, \Cref{clm:submod} holds because $\lpsubmod$ can be viewed as the solution of a fractional knapsack problem, and function $f_t(\{r_{it}\})$ can be viewed as taking the concave closure of some unit-demand function. We defer the proof of \Cref{clm:submod} to the end of this subsection.

Now, we show how to understand the Step $4$ of our algorithm as an online submodular welfare maximization process. For each $i \in [|U|]$ and $t:(i,t) \notin L$, define $\hat r_{it} = x^{(t_i)}_{it}$, i.e., the initial value we set for deterministic edges when $i$ is transitioned from stage $1$ to stage $2$. For $j \in \{0,1, ..., |U|\}$, further define $\hat r^{(j)}_{it} = \one[j \geq i] \cdot \hat r_{it}$. Then, the Step $4$ of our algorithm can be viewed as setting $\{\hat r_{jt}\}$ subject to the constraint $\sum_{t:(j,t) \notin L} \hat r_{jt} \leq 1 - \probfr$ for each $j = 1, 2, \cdots, |U|$, such that the marginal gain $\sum_{t \in V} f_t(\{\hat r^{(j)}_{it}\}) - \sum_{t \in V} f_t(\{\hat r^{(j-1)}_{it}\})$
is maximized. 

Define $\tilde r_{it} = (1 - \probfr) \cdot y^*_{it}$ for $(i,t) \notin L$. Note that  $\{\tilde r_{it}\}$ also satisfies the constraint that $\sum_{t:(i,t) \notin L} \tilde r_{it} \leq 1 - \probfr$. Since Step $4$ of our algorithm is a greedy approach that maximizes the marginal gain, the total welfare of $\{\hat r_{it}\}$ should be at least one half times the total welfare of $\{\tilde r_{it}\}$. To be specific, define $\bar r_{it} = \max\{\hat r_{it}, \tilde r_{it}\}$. For each $t \in V$, we have
\begin{align*}
     f_t(\{\tilde r_{it}\}) -  f_t(\{\hat r_{it}\}) ~\leq&~  f_t(\{\bar r_{it}\}) - f_t(\{\hat r_{it}\}) \\
    ~=&~ \sum_{j = 1}^{|U|} \left( f_t\big(\{\bar r_{it}\}_{i < j} \cup \{\bar r_{jt}\} \cup \{\hat r_{it}\}_{i > j}\big) - f_t\big(\{\bar r_{it}\}_{i < j} \cup \{\hat r_{jt}\} \cup \{\hat r_{it}\}_{i > j}\big) \right) \\
    ~\leq&~ \sum_{j = 1}^{|U|} \left( f_t\big(\{\hat r_{it}\}_{i < j} \cup \{\bar r_{jt}\} \cup \{0\}_{i > j}\big) - f_t\big(\{\hat r_{it}\}_{i < j} \cup \{\hat r_{jt}\} \cup \{0\}_{i > j}\big) \right) \\
    ~\leq&~ \sum_{j = 1}^{|U|} \left( f_t\big(\{\hat r_{it}\}_{i < j} \cup \{\tilde r_{jt}\} \cup \{0\}_{i > j}\big) - f_t\big(\{\hat r_{it}\}_{i < j} \cup \{r_{jt} = 0\} \cup \{0\}_{i > j}\big) \right), 
\end{align*}
where the second line is the telescoping sum, the third line uses the submodularity (\Cref{clm:submod}) and the fact that $\hat r_{it} \leq \bar r_{it}$, and the the last line uses the fact that $\bar r_{jt} - \hat r_{jt} = (\tilde r_{jt} - \hat r_{jt})^+ \leq \tilde r_{jt}$, and the submodularity. Summing the above inequality for $t \in V$, we have
\begin{align*}
    \sum_{t \in V} f_t(\{\tilde r_{it}\}) - \sum_{t \in V} f_t(\{\hat r_{it}\}) ~&\leq~ \sum_{j = 1}^{|U|} \left(\sum_{t \in V} f_t(\{\hat r^{(j-1)}_{it}\}_{-j} \cup \{\tilde r_{jt}\}\}) - \sum_{t \in V} f_t(\{\hat r^{(j-1)}_{it}\})\right) \\
    ~&\leq~ \sum_{j = 1}^{|U|} \left(\sum_{t \in V} f_t(\{\hat r^{(j-1)}_{it}\}_{-j} \cup \{\hat r_{jt}\}\}) - \sum_{t \in V} f_t(\{\hat r^{(j-1)}_{it}\})\right) \\
    ~&=~ \sum_{j = 1}^{|U|} \left(\sum_{t \in V} f_t(\{\hat r^{(j)}_{it}\}) - \sum_{t \in V} f_t(\{\hat r^{(j-1)}_{it}\})\right) ~=~\sum_{t \in V} f_t(\{\hat r_{it}\}),
\end{align*}
where the second line uses the fact that $\{\hat r_{jt}\}$ maximizes the marginal gain of $\sum_{t \in V} f_t(\{\hat r_{it}\})$. Rearrange the above inequality, we have
\begin{align}
    \sum_{t \in V} f_t(\{\hat r_{it}\}) ~\geq~ \frac{1}{2} \sum_{t \in V} f_t(\{\tilde r_{it}\}). \label{eq:submod-main}
\end{align}

It remains to further bound the values of $\sum_{t \in V} f_t(\{\hat r_{it}\})$ and $\sum_{t \in V} f_t(\{\tilde r_{it}\})$. For $\sum_{t \in V} f_t(\{\hat r_{it}\})$, we note the following two properties of our algorithm:
\begin{itemize}
    \item At each $t \in V$, the value of $\{x^{(t)}_{is}\}$ records the optimal solution of $\lpsubmod$ with respect to the current $\sum_{t \in V} f_t(\{\hat r^{(j)}_{it}\})$, where $j$ is the last vertex which is transitioned to stage $2$ before $t$. This is true by observing that the Step $4$ of our algorithm is a greedy algorithm that maintains the highest weight edges for $\{\hat r^{(j)}_{it}\}$.
    \item The value of $f_t(\{\hat r_{it}\})$ is settled when online vertex $t$ arrives. This is true because for every $j$ which is still in stage $1$ (i.e., $t_j \geq t$), there must be $\widehat w_{jt} = 0$. Therefore, the arrival of $j$ can't change the value of $f_t(\{\hat r_{it}\})$.
\end{itemize}
With the above two observations, we have
\begin{align}
    \sum_{t \in V} f_t(\{\hat r_{it}\})  ~=~ \sum_{(i,t) \in E} \widehat w_{it} \cdot x^{(t)}_{it} - \sum_{(i,t) \in L} \widehat w_{it} \cdot \tildxl_{it}. \label{eq:submod-hat}
\end{align}

For $\sum_{t \in V} f_t(\{\tilde r_{it}\})$, note that the following is a feasible solution of $\lpsubmod$ with respect to $\{\tilde r_{it}\}$:
\[
z_{it} = \begin{cases} \min(\tildxl_{it}, \yit^*) & (i,t) \in L  \\ (1-\probfr) \cdot \yit^* & (i,t) \in E_2 \setminus L \\ 0 & (i,t) \in E_1 \setminus L
\end{cases}~
\]
Therefore, we have
\begin{align}
    \sum_{t \in V} f_t(\{\tilde r_{it}\}) ~&\geq~ (1 - \probfr) \cdot \sum_{(i,t) \in E_2 \setminus L} \widehat w_{it} \cdot y^*_{it} + \sum_{(i,t) \in L} \widehat w_{it} \cdot \left(\min(\tildxl_{it}, \yit^*) - \tildxl_{it}\right) \notag \\
    ~&\geq~ (1 - \probfr) \cdot \sum_{(i,t) \in E_2 \setminus L} \widehat w_{it} \cdot y^*_{it} - \sum_{(i,t) \in L} \widehat w_{it} \cdot \left(\tildxl_{it}- \yit^*\right)^+ \label{eq:submod-tilde}
\end{align}

Applying \eqref{eq:submod-hat} and \eqref{eq:submod-tilde} to \eqref{eq:submod-main} proves \Cref{lma:submod}.

Finally, we finish the proof of \Cref{lma:submod} by presenting the missing proof of \Cref{clm:submod}.

\paragraph{Proof of \Cref{clm:submod}}
It's sufficient to consider the case where $\Delta_r = \delta_r \to 0$. For the general $\Delta_r$, applying the inequality for $\delta_r \to 0$ together with an integral is sufficient.

    Let ${z^{it}}$ be the optimal solution of $\lpsubmod$ with respect to ${r{it}}$. Now, consider increasing the value of $r_{jt}$ by $\delta_r$. The following greedy approach updates ${z^{it}}$ to maintain its optimality for $\lpsubmod$ after $r_{jt}$ is incremented by $\delta_r$:
    \begin{itemize}
        \item Case 1: If $\sum_{i \in U} z^*_{it} < p_t$, increase $z^*_{jt}$ by $\delta_r$.
        \item Case 2: Otherwise, let $\ell = \arg \min_{i: z^*_{it} > 0} \widehat w_{it}$. Decrease $z^*_{\ell t}$ by $\delta_r$ and increase $z^*_{jt}$ by $\delta_r$ if $\widehat w_{jt} > \widehat w_{\ell t}$.
    \end{itemize}
    Recall that we assume $\delta_r \to 0$. Therefore, the feasibility of above approach is guaranteed by the fact that we have $\sum_{i \in U} z^*_{it} + \delta_r \leq p_t$ for Case 1 and $z^*_{\ell t} - \delta_r \geq 0$ for Case 2 when $\delta_r$ is sufficiently small. 
    
    Define 
    \[
    \tau = \one\left[\sum_{i \in U}z^*_{it} = p_t\right] \cdot \min_{i: z^*_{it} > 0} \widehat w_{it}.
    \]
    to be the ``threshold value'' of the above greedy process. Note that in Case 1, the value of $\tau$ is $0$ before making the update, while in Case 2, we only decrease the value of the LP variable with a minimum weight. Therefore, the value of $\tau$ is non-decreasing during the above greedy process.

    Let $\tau_r$ be the threshold value $\tau$ with respect to solution $\{r_{it}\}$. Starting from $\{r_{it}\}$, consider to repeatedly perform the above approach for each coordinate of $\{r_{it}\}$ and maintain the optimality of solution $\{z^*_{it}\}$, until $\{r_{it}\}$ is updated to $\{r'_{it}\}$. Let $\tau_{r'}$ be the corresponding threshold value with respect to the converted solution $\{r'_{it}\}$. Then, we have $\tau_r \leq \tau_{r'}$.

    We finish the proof by noting that when $\delta_r \to 0$, we have
    \begin{align*}
        f_t\left(\{r_{it}\}_{-j} \cup \{r_{jt} + \Delta_r\}\} \right) - f_t\left(\{r_{it}\} \right) ~&=~ (\widehat w_{jt} - \min\{\widehat w_{jt}, \tau_r\}) \cdot \delta_r \\
        ~&\geq~ (\widehat w_{jt} - \min\{\widehat w_{jt}, \tau_{r'}\}) \cdot \delta_r \\
        ~&=~f_t\left(\{r'_{it}\}_{-j} \cup \{r'_{jt} + \Delta_r\}\} \right) - f_t\left(\{r'_{it}\} \right),
    \end{align*}
    where the two equalities follow from the fact that $(\widehat w_{jt} - \min\{\widehat w_{jt}, \tau_{r}\}) \cdot \delta_r$ and $(\widehat w_{jt} - \min\{\widehat w_{jt}, \tau_{r'}\}) \cdot \delta_r$ are the marginal gains of the greedy process for solutions $\{r_{it}\}$ and $\{r'_{it}\}$ respectively. 
\end{proof}